\documentclass[11pt,notitlepage]{article}
\usepackage[margin=1in]{geometry}
\usepackage[utf8]{inputenc}
\usepackage[thmnum]{preambletemp}
\usepackage{authblk}

\newcommand{\problemDefn}{\hyperref[defn:problem_defn]{Robust Multivariate Polynomial Regression Problem}\xspace}

\title{Outlier Robust Multivariate Polynomial Regression}

\author[1]{Vipul Arora\thanks{Supported in part by NRF-AI Fellowship R-252-100-B13-281.}}
\affil{\footnotesize National University of Singapore. $\{$\texttt{\href{mailto:vipul@comp.nus.edu.sg}{vipul},\href{mailto:arnab@comp.nus.edu.sg}{arnab}}$\}$\texttt{@comp.nus.edu.sg}.}
\author[1]{Arnab Bhattacharyya\thanks{Supported in part by NRF-AI Fellowship R-252-100-B13-281, Amazon Faculty Research Award, and Google South \& Southeast Asia Research Award.}}
\author[2]{Mathews Boban\thanks{Supported in part by NRF-AI Fellowship R-252-100-B13-281.}}
\affil{\footnotesize National University of Singapore. \texttt{\href{mailto:mathewsboban242@gmail.com}{mathewsboban242@gmail.com}}}
\author[3]{Venkatesan Guruswami\thanks{Supported in part by NSF CCF-2211972 and a Simons Investigator Award}}
\affil{\footnotesize University of California, Berkeley. \texttt{\href{mailto:venkatg@berkeley.edu}{venkatg@berkeley.edu}}.}
\author[4]{Esty Kelman\thanks{Supported in part by an Amazon Faculty Research Award to AB, in part by ERC grant 834735, and in part by NSF TRIPODS program (award DMS-2022448) }}
\affil{\footnotesize Massachusetts Institute of Technology, and Boston University. \texttt{\href{mailto:ekelman@mit.edu}{ekelman@mit.edu}}.}

\date{February 2024}
\begin{document} 
\parskip=0.5ex
\maketitle

\thispagestyle{empty}

\allowdisplaybreaks

\begin{abstract}
We study the problem of \emph{robust multivariate polynomial regression}: 
let $p\colon\mathbb{R}^n\to\mathbb{R}$ be an unknown $n$-variate polynomial of degree at most $d$ in each variable. We are given as input a set of random samples $(\mathbf{x}_i,y_i) \in [-1,1]^n \times \mathbb{R}$  that are noisy versions of $(\mathbf{x}_i,p(\mathbf{x}_i))$. More precisely, each $\mathbf{x}_i$ is sampled independently from some distribution $\chi$ on $[-1,1]^n$, and for each $i$ independently, $y_i$ is arbitrary (i.e., an outlier) with probability at most $\rho < 1/2$, and  otherwise satisfies $|y_i-p(\mathbf{x}_i)|\leq\sigma$. The goal is to output a polynomial $\hat{p}$, of degree at most $d$ in each variable, within an $\ell_\infty$-distance of at most $O(\sigma)$ from $p$.

\smallskip
Kane, Karmalkar, and Price [FOCS'17] solved this problem for $n=1$. We generalize their results to the $n$-variate setting, showing an algorithm that achieves a sample complexity of $O_n(d^n\log d)$, where the hidden constant depends on $n$, if $\chi$ is the $n$-dimensional Chebyshev distribution. The sample complexity is $O_n(d^{2n}\log d)$, if the samples are drawn from the uniform distribution instead. The approximation error is guaranteed to be at most $O(\sigma)$, and
the run-time depends on $\log(1/\sigma)$. In the setting where each $\mathbf{x}_i$ and $y_i$ are known up to $N$ bits of precision, the run-time's dependence on $N$ is linear.  We also show that our sample complexities are optimal in terms of $d^n$. 
Furthermore, we show that it is possible to have the run-time be independent of $1/\sigma$, at the cost of a higher sample complexity.

\end{abstract}

\clearpage
\section{Introduction}\label{sec:intro}
``Curve fitting'' or {\em polynomial regression} is one of the oldest and most fundamental learning problems: find a polynomial that approximately satisfies the input-output relationship displayed by a collection of data points. Polynomial regression has a vast range of applications, from the physical sciences to statistics and machine learning; see, e.g., the books \cite{wolberg2006data, zielesny2011curve} for discussions and references. 

The focus of this work is on {\em multivariate} polynomial regression, 
which is the task of learning the class of bounded degree polynomials from random noisy samples. Multivariate polynomial regression is a natural requirement in many applications. For example, in computer vision, boundaries of objects are often modeled as low-degree bivariate polynomials, so it is well-motivated to fit curves to estimates of object boundaries.
Our goal is to design {\em robust} regression algorithms, which can withstand having a constant fraction of the input data be arbitrary outliers in the same setting as in \cite{arora-khot,guruswami-zuckerman, kkp}.

We next formally state the problem of robust multivariate regression. Let us denote by $\Pd$ the class of all $n$-variate \emph{individual} degree-$d$ polynomials, which are the polynomials with degree at most $d$ in each variable\footnote{This is in contrast to the usual convention of the \emph{total} degree being at most $d$. Note that the class of polynomials of \emph{total} degree at most $d$ is strictly included in $\Pd$. Our results (for $\Pd$) can be translated for the class of total degree-$d$ polynomials; See discussion in \autoref{rem:indiv-to-total}.}.  
\begin{restatable}{problem_defn}{problemdefn}\label{defn:problem_defn}

    Let $\sigma>0$ be a noise bound, $C> 1$ be an approximation factor, $\rho \in [0,1]$ be the outlier probability, $\chi$ be a probability distribution over  $[-1,1]^n$.  Fix an unknown $p \in \Pd$ and let $S=\{(\mathbf{x}_1, y_1),\dots,(\mathbf{x}_M, y_M)\}$ be a set random samples where for each $i\in[M]$ independently, $\mathbf{x}_i$ sampled from $\chi$, and $y_i\in\R$ is an \emph{inlier} satisfying $|{y}_i-p(\mathbf{x}_i)|\leq\sigma$ with probability $1-\rho$ and otherwise, it may be an \emph{outlier}, i.e., the noise may be arbitrarily large.
 The goal is to design an efficient algorithm that, given the set $S$ of random samples as input, recovers a polynomial $\hat{p}\in \Pd$ satisfying
\[  \max_{\mathbf{x}\in [-1,1]^n} |p(\mathbf{x})-\hat{p}(\bfx)|\leq C\sigma,\] 
with probability at least $1-\delta$.
   \end{restatable}

Note that though the locations of the outliers are random, i.e., each sample is an {outlier} with probability $\rho$ independently, the noise for both the {inliers} and the {outliers} is still allowed to be chosen in an adversarial way (meaning an adversary can choose the values of all the $y_i$'s after seeing the entire sample set $\{\bfx_i\}$). 

In the univariate setting, recovery for non-trivial values of $\rho$ was first shown by Guruswami and Zuckerman \cite{guruswami-zuckerman} for $\rho < 1/\log d$. Previously, Arora and Khot \cite{arora-khot} had shown $\rho<1/2$ was information-theoretically necessary for unique recovery. Subsequently, Kane, Karmalkar, and Price \cite{kkp} designed a simple and optimal (up to constants) algorithm that runs in polynomial time for any $\rho < 1/2$, uses  $\Theta(d\log d)$ samples from the Chebyshev measure on $[-1,1]$, or $\Theta(d^2)$ uniform samples 
and outputs a degree-$d$ univariate polynomial $\hat{p}$ satisfying $\max_{x \in [-1,1]} |p(x)-\hat{p}(x)| \leq C\sigma$. They show how to achieve $C$ as close to $2$ as desired. In addition, they show that to solve the problem for $d=2$ with probability at least $2/3$, $C>1.09$ is needed, while for general $d$, to succeed with constant probability, one needs $C>1+\Omega(1/d^3)$.

\subsection{Main results}

 We wish to minimize the sample complexity $M$. 
  Our algorithmic results are mainly when the measure $\chi$  is either the uniform distribution or the $n$-dimensional Chebyshev measure, i.e., the $n$-fold product of the Chebyshev measure on $[-1,1]$, with the probability density function $\propto1/\sqrt{1-x^2}$ for $x\in[-1,1]$.

Note that when $n$ is large, for some distributions, solving the multivariate polynomial regression problem requires $\exp(n)$ many samples, even for polynomials of \emph{total} degree $d=1$ (for completeness, we provide a proof in \autoref{subsec:lower-bound-linear}). So, for sample-efficient algorithms, it is prudent to assume $n>1$ being a constant. In this setup, then, the total degree of an individual degree-$d$ polynomial is at most $nd$, i.e., $O(d)$, and hence the multivariate polynomial regression problem becomes oblivious to the degree being total or individual. Thus, we focus on learning the class $\Pd$ of \emph{individual} degree-$d$ polynomials in a constant number of variables.

We now state our main results. Denote the cube $[-1,1]^n$ by $\cube_n$; we will omit the subscript when the dimension is clear from the context. Let $\|\cdot\|_{\cube, \infty}$ denote the $\ell_\infty$ norm over $\cube_n$.

\begin{restatable}{thm}{linftyregressionupperboundintro} \label{thm:l_infty_regression_upper_bound_intro}
     Let $\sigma\geq 0,\adderr>0$, and $\rho$ be any constant $<1/2$. 
   There is an algorithm that almost solves the \problemDefn with a constant approximation factor,  up to an additive error of $\adderr$. The output of the algorithm is $\hat{p}\in \Pd$ that satisfies
   \[|p(\bfx)-\hat{p}(\bfx)|\leq O(\sigma)+\eta,\qquad \text{ for all } \bfx\in\cube_n,\]
with probability at least $2/3$. It uses   
 $M=O_n(d^n\log d)$ samples drawn from the multidimensional Chebyshev distribution, or $M=\widetilde{O}_{n}(d^{2n})$ if the samples are drawn from the  uniform measure.
 Its run-time is at most $\poly(\log\norm{p}_{\cube,\infty},M,\log(1/\eta))$.
\end{restatable}
The notations $\widetilde{O}$, $\widetilde{\Theta}$  hide factors proportional to $\log d$ above; the dependence on $\eta$ or $\sigma$ is kept explicit.

In case $\sigma$ is known to be at least $ 2^{-N}$, one may choose $\eta=2^{-N}$  to guarantee $\|\hat{p}-p\|_{\cube_n,\infty}\le O(\sigma)$ and run-time proportional to $\poly(N)$. Generalizing this observation, we consider the {\em $N$-bit precision setting}, where both the sample locations $\bfx_i$ and the labels $y_i$ are truncated to $N$ bits of precision; this is consistent with a computational model where real numbers can only be specified up to $N$ bits of precision.  We  show that in the $N$-bit precision setting, a variant of our algorithm achieves a constant approximation factor without any additional additive error. 
 \begin{restatable}{thm}{finitepresicionintro} \label{thm:finite_presicion_intro}
     Let $N$ be the number of bits of precision,  $\sigma\ge 2^{-N}$, and 
$\rho$ be any constant $<1/2$. 
   There exists an algorithm for the \problemDefn, wherein each $\bfx_i$  is now drawn from a continuous distribution $\chi$ and then rounded to $N$ bits of precision, and each $y_i$ is similarly rounded. The  output of the algorithm is $\hat{p}\in \Pd$, that satisfies  
   \[|p(\bfx)-\hat{p}(\bfx)|\leq O(\sigma), \qquad \text{ for all } \bfx\in\cube_n,\]
    with probability at least  $2/3$.
 It uses   
 $M=O_n(d^n \log d)$ samples drawn from the  multidimensional Chebyshev distribution, or $M=\widetilde{O}_{n}(d^{2n})$ if the samples are drawn from  the uniform measure. 
 Its run-time is at most $\poly(\log\norm{p}_{\cube,\infty},M,N)$.
\end{restatable}

To avoid a run-time dependent on $\norm{p}_{\cube,\infty}$ and $1/\eta$, in case they are unknown or too large, we also obtain a variant of the algorithm that achieves an explicit constant multiplicative approximation factor, as close to $2$ as desired and independent of $\norm{p}_{\cube, \infty}$ and $1/\eta$, at the cost of a higher sample complexity. 
\begin{restatable}{thm}{linftywithloneupperboundintro}\label{cor:final-cor-with-l-1-intro}
 Let $\eps>0$, $\sigma\geq 0$, and a constant $\rho<1/2$.
   There exists an algorithm that solves the \problemDefn. The  output of the algorithm is $\hat{p}\in \Pd$, that satisfies    
   \[|p(\bfx)-\hat{p}(\bfx)|\leq (2+\eps)\sigma\qquad \text{ for all } \bfx\in\cube_n,\]
    with probability at least  $2/3$.
It uses $M=\poly(d^{n^2},1/\eps^n)$ samples drawn from either the multidimensional Chebyshev distribution or the uniform distribution.
 Its run-time is $\poly(M)$.
\end{restatable}

We complement the above results by showing lower bounds on the sample complexity of robust multivariate polynomial regression. 

\begin{restatable}{thm}{lowerboundintro}\label{thm:lower-bound-intro}
For any approximation factor $C>1$, there exists  $c=c(C)>0$ such that any algorithm, that solves the \problemDefn, requires at least $(cd)^{2n}$ samples drawn from the uniform measure to succeed with probability more than  $2/3$.
This holds for any outlier probability $\rho$.  
\end{restatable}

The lower bound matches the upper bound of \autoref{thm:finite_presicion_intro}  up to lower order terms (in the case of uniform sampling) for constant $n$, and $C$, and holds even for $\rho=0$, where there are no outliers. The following result shows that our result in \autoref{thm:finite_presicion_intro} for the multidimensional Chebyshev measure matches the optimal sample complexity over arbitrary distributions\footnote{Similarly, the lower bounds match the respective sample complexities of \autoref{thm:l_infty_regression_upper_bound_intro}, when the additive error $\eta$ approaches $0$, since the algorithm's run-time grows as $\eta\to 0$, but the sample complexities remain unchanged.}.

\begin{restatable}{thm}{distfreelowerboundintro}\label{thm:dist-free-lower-bound-intro}
For any approximation factor $C>1$, and any  outlier probability $\rho>0$, there exists $c=c(C,\rho)>0$ such that any algorithm, that solves the \problemDefn, requires at least ${(cd)^{n}\log d}$ samples drawn from any measure over $[-1,1]^n$ to succeed with probability more than $2/3$.
\end{restatable}

 A comparison of the results across parameter regimes may be given via the following table, wherein $M$ is the sample complexity, and $C_p=\log \|p\|_{\cube,\infty}$:
\begin{table}[h]
    \centering
    \begin{tabular}{|c|c|c|c|c|}
    \hline
         Setting & Approximation & Chebyshev Measure & Uniform Measure & Run-time \\
         \hline
         Exact & $O(\sigma)+\eta$ & $O_n(d^n \log d)$ & $\widetilde{O}_n(d^{2n})$ & $\poly(C_p,M,\log(1/\eta))$ \\
         $N$-bit & $O(\sigma)$ & $\Theta_n(d^n\log d)$ & $\widetilde{\Theta}_n(d^{2n})$ & $\poly(C_p,M,N)$ \\
         Small $\eps$ & $(2+\eps)\sigma$ & $\poly(d^{n^2},1/\eps^n)$ & $\poly(d^{n^2},1/\eps^n)$ & $\poly(M)$\\
         \hline
    \end{tabular}\caption{The upper bounds in the first line follow from \autoref{thm:l_infty_regression_upper_bound_intro}, the second line from \autoref{thm:finite_presicion_intro}, and the third line from \autoref{cor:final-cor-with-l-1-intro}. The lower bounds in the second line for the Uniform Measure follow from \autoref{thm:lower-bound-intro}, and for the Chebyshev Measure from \autoref{thm:dist-free-lower-bound-intro}.}
    \label{tab:comparison}
\end{table}
\begin{remark}\label{rem:kkp-highlight}
\cite{kkp} were able to achieve both optimal sample complexity and efficient run-time independent of $\norm{p}_\infty/\sigma$ with a single algorithm in the univariate setting. In contrast, as we elaborate in the proof overview, our \autoref{cor:final-cor-with-l-1-intro} incurs an additional blowup in the sample complexity; we leave open the problem of realizing the error guarantees of \autoref{cor:final-cor-with-l-1-intro} with the optimal number of samples.
\ignore{
    KKP didn't highlight the result corresponding to \autoref{thm:l_infty_regression_upper_bound_intro} separately. 
    They successfully showed how to achieve both optimal sample complexity and efficient running time independent of $\norm{p}_\infty/\sigma$ in a single algorithm by reusing the same set of samples for both the steps, corresponding to \autoref{cor:final-cor-with-l-1-intro} and \autoref{thm:l_infty_regression_upper_bound_intro}.  In contrast, we could not achieve the optimal sample complexity for the additional step of $\ell_1$ regression (we elaborate on this step in the proof overview) that is required for \autoref{cor:final-cor-with-l-1-intro}.}
\end{remark}

\begin{remark}\label{rem:indiv-to-total}
     Both our upper and lower bounds hold for the class of total degree-$d$ polynomials, when $n$ is bounded, since total degree being at most $d$ implies individual degree being at most $d$, and individual degree being at most $d$ implies total degree being at most $dn$.
\end{remark}

\subsection{Main technical contributions}
Our main technical contributions are twofold, and they may be of interest more broadly.
First, let $\{\cube_{\bm j}\}_{\bm j\in[m]^n}$ be a partition of the cube $\cube_n$ induced by the $m$-Chebyshev extremas on each axis. We call it the $(m,n)$-Chebyshev partition\footnote{See formal \autoref{defn:Chebyshev-cells}, and \autoref{fig:2-d-grid} for an illustration of a $2$-dimensional Chebyshev partition.} of $\cube_n$. 
\begin{figure}[ht]
    \centering
    \includegraphics[scale=0.5]{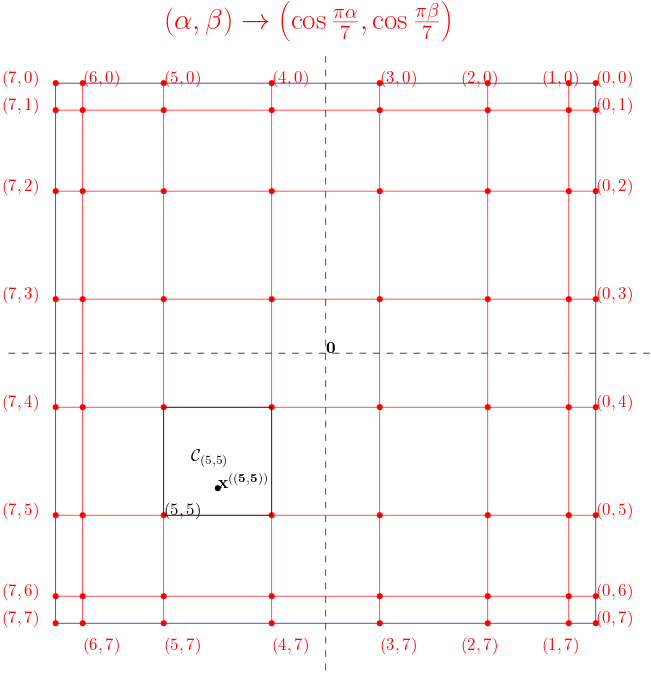}
    \caption{An illustration of a 2-dimensional $(7,2)$-Chebyshev partition (in red) super-imposed on the 2-dimensional solid cube $\cube_2=[-1,1]^2$, with boundary in blue. The cells are indexed by their bottom-left Chebyshev extremas (the red points). \autoref{thm:optimal-l-infty-bound} essentially proves that on any cell, for example, $\cube_{(5,5)}$ (in black), $p$ can be well approximated by its evaluation on one arbitrary point $\bfx^{((5,5))}\in\cube_{(5,5)}$.}
    \label{fig:2-d-grid}
\end{figure}
For $n=1$, \cite{kkp} showed how to approximate a univariate polynomial of degree at most $d$ on $[-1,1]$ by an appropriate piece-wise constant function with respect to the Chebyshev partition. We extend their result to the multivariate case. 

\begin{restatable}{thm}{optimallinftybound}[Approximation by piece-wise constant functions]\label{thm:optimal-l-infty-bound}
    Let $p:\cube_n\to\R$ be a polynomial of degree at most $d$ in each variable, and  $m\geq d$.  
    Let $r:\cube_n\to\R$ be a piece-wise constant function with respect to  the $(m,n)$-Chebyshev  partition, such that for every  $\bm j\in[m]^n$, there exist a point $ \mathbf{x}^{(\bm j)}\in \cube_{\bmj}$, such that $r(\mathbf{x})=p(\mathbf{x}^{(\bm j)})$, for all $\bfx\in \cube_{\bmj}$. Then, there exists a universal constant $C$ such that,
    \[\|p-r\|_{\cube_n,\infty}\leq C\frac{dn}{m} \|p\|_{\cube_n,\infty}.\]
\end{restatable}
This is proved in \autoref{subsec:optimal-l-infty-bound}.
Second, we show how to relate the maximum value of a bounded degree polynomial on $\cube_n$ with its $\ell_1$ norm, on the same cube $\cube_n$.
\begin{restatable}{thm}{linftytolone}\label{thm:improved-sandwich}
    There exists a global constant $C>0$ such that, for every polynomial $p\in \Pd$: 
     $$ \|p\|_{\cube_n,\infty}\leq C^n d^{2n}\|p\|_{\cube_n,1}.$$
\end{restatable}
We also prove the tightness of the relation between $\ell_\infty$ and $\ell_1$ norms.
\begin{restatable}{prop}{tightnessofinduction}\label{lem:tightness_of_induction}
There exists a global constant $c>0$ such that for
    every odd $d$, there exists a family of individual degree-$d$ polynomials $\{f_n\}_{n\in\mathbb{N}}$, where $f_n:\cube_n\to\R$, such that $\|f_n\|_{\cube_n,\infty} \geq c^nd^{2n} \|f_n\|_{\cube_n,1}$.
\end{restatable}
These are proved in \autoref{sec:l-infty-to-l-1}.
In \autoref{sec:wilhwlmwnsen}, we show how to prove a similar result, using a previous work by \cite{WILHELMSEN1974216} and a simpler argument, but it has
 factors that are worse by a multiple of $\sim n^{2.5n}$.

\subsection{Related Work}
Given the fundamental nature of the polynomial regression problem, there is a long history of work on the problem, but mostly in the univariate setting. Arora and Khot \cite{arora-khot} were the first to study this problem in our random outlier noise model, giving an algorithm that in $O(\frac{d^2}{\sigma}\log\frac{d}{\sigma})$ random noisy samples outputs an $O(\sigma)$-approximation (in $\ell_\infty$) to the (actual) hidden polynomial, where the outlier rate $\rho=0$. This was improved in a work by Guruswami and Zuckerman \cite{guruswami-zuckerman},
who gave a computationally efficient algorithm for all $\rho < 1/\log d$. Finally, in a significant improvement, Kane, Karmalkar and Price \cite{kkp} obtained computationally efficient algorithms for any $\rho < 1/2$, while having no additional requirements for $\sigma$ or $\|p\|_\infty$. 
As far as we know, Daltrophe, Dolev and Lotker \cite{DDL2018} were the first to consider the \emph{multivariate} setting of this problem. For the two-dimensional case ($n=2$), they gave an algorithm that with $O(\frac{d^4}{\sigma}\log \frac{d}{\sigma})$ random noisy samples outputs a $c(2)\cdot \sigma$-approximation (in $\ell_\infty$), for any $\rho <\frac{1}{2}$, where $c(2)=3$. A limitation of their result is that $c(n)$ grows exponentially in $n$. In contrast, we obtain a constant factor approximation for all $n$.

There has also been a surge of recent research in the related area of robust statistics. Here, instead of the outliers being randomly placed, their locations are chosen adversarially. For the setting when the \emph{total} degree is fixed, and the dimension $n$ is growing, Klivans, Kothari and Meka
\cite{KKM2018} gave an algorithm using the sum-of-squares method. However, their sample complexity is $\mathrm{poly}(n^d)$, which is exponential in the degree; moreover, the output guarantee is with respect to the $\|\cdot\|_2$ norm, instead of the $\|\cdot\|_\infty$ norm in our setting. Other related works in this spirit are that of Diakonikolas, Kamath, Kane, Li, Steinhardt and Stewart \cite{diakonikolas2019sever} and Prasad, Suggala, Balakrishnan and Ravikumar \cite{prasad2020robust}. The work of Diakonikolas, Kong, and Stewart \cite{diakonikolas2019efficient} also studied the related problem of adversarially robust linear regression, but with the assumption that the $\bfx_i$'s are drawn from a Gaussian.

\subsection{Technical Overview}
We first sketch the algorithm designed by Kane, Karmalkar, and  Price \cite{kkp}, henceforth KKP, and their analysis for the univariate case, $n=1$.
For univariate polynomial interpolation, the points at which the noisy samples are located play an important role in determining the interpolation error. 
Choosing the points to be the Chebyshev nodes, which are the roots of Chebyshev polynomials (see \autoref{defn:chebyshev-polynomial}),
is a good starting point, as suggested by approximation theory literature. However, the algorithm receives random samples, which may not necessarily be located at the Chebyshev nodes. Instead, KKP argue that they have enough inliers around each Chebyshev node. 
For this, they define a partition of the interval $[-1,1]$ on the {extremal} points of Chebyshev polynomials, which they call \emph{the size-$m$ Chebyshev partition}. 
In their algorithm and its analysis, they assume that the set of samples is \emph{good}, in the sense that in every part of the partition, there is only a small fraction of outliers; this good event is guaranteed to happen with high probability.  

\subsubsection{KKP's Algorithm and Analysis}
Formally, the {size-$m$ Chebyshev partition} of $[-1,1]$ is the set of intervals $I_j=\left[\cos\frac{\pi j}{m},\cos\frac{\pi(j-1)}{m}\right]$, for all $j\in[m]$. 
Given a set of $s$ samples $(x_i,y_i)$ where $x_i$'s are drawn from some distribution over $[-1,1]$, and $y_i$'s are the corresponding labels, the algorithm uses the idea of  \emph{median-based recovery}. For every interval $I_j$: \begin{itemize}
\itemsep=0ex
    \item Let $\tilde{y}_j$ be the median of $y_i$'s of samples for which  $x_i\in I_j$. Since the set of samples is assumed to be \emph{good}, i.e., the fraction of outliers in each interval is strictly less than $1/2$, $\tilde{y}_j$ lies in between two inliers located in the interval $I_j$.

   \item Let $\tilde{x}_j$ be an arbitrary point in $I_j$.
            \item Let $\hat{p}$ be a minimizer, over all degree-$d$ polynomials, of the empirical $\ell_\infty$ error $\max_j |\hat{p}(\tilde{x}_j)-\Tilde{y}_j|$  over all $j\in[m]$.
\end{itemize} 
As $m$ grows, the partition gets finer, and the error gets better, though at a cost of higher sample complexity. Iteratively applying the median-based recovery on the residual left from previous iteration improves the approximation, and in $\log(\max_{x\in[-1,1]}|p(x)|/\sigma)$ iterations, the error drops down to $3\sigma$. 

\ignore{We note that KKP reuses the same set of samples for all iterations, so that the number of iterations doesn't affect their sample complexity. For a result $\hat{p}$ of some iteration, the samples for the next iteration are now the set $\{(x_i,y_i-\hat{p}(x_i))\}$.
In addition, the sample complexity needed for the first step of $\ell_1$ regression is of the same order as the sample complexity needed for the iteratively applying the median-based recovery on residuals.  }

The backbone of their analysis is a {technical} result for approximating $p$ on a size-$m$ Chebyshev partition, by a piece-wise constant function (with respect to the same partition) that matches $p$ on at least one point in every part of the partition.
\begin{restatable}{lem}{kkptechnicallemma}[Lemma 2.1, \cite{kkp}]\label{lem:kkp-technical-lemma}
    Let $g:\R\to\R$ be a (univariate) degree-$d$ polynomial. Let $\{I_j\}_{j\in[m]}$ denote the $m$-Chebyshev partition of $[-1,1]$, for some $m\geq d$. Let $r:[-1,1]\to\R$ be piece-wise constant, so that for each $k\in[m]$, there exists $x_k^\ast\in I_k$, such that $r(x)=g(x_k^\ast)$ for all $x\in I_k$. Then, there exists a universal constant $C$ such that, for any $q\geq 1$,
    \[\|g-r\|_{q}\leq\frac{Cd}{m}\|g\|_q.\]
\end{restatable}
To prove \autoref{lem:kkp-technical-lemma}, they used Nevai's inequality \cite{nevai1979bernstein}, an $\ell_q$-version of Bernstein's inequality, to bound the $\ell_q$ approximation error by a multiple of the $\ell_q$ norm of $p$. The multiple is linear in the degree $d$, and $1/m$.
The bound from Nevai's inequality works for all ``inner" parts of the Chebyshev partition, as it relies on the fact that the length of any part $I_j$, where $j\not\in \{1,m\}$, is at most $O(\sqrt{1-x^2}/m)$, for every $x\in I_j$. To bound the approximation error on the peripheral parts $I_1,I_m$, they use Markov Brothers' inequality (\autoref{lem:markovbros}). Here they strongly rely on the fact that those parts are much narrower\footnote{A pictorial demonstration of this narrowness, for 2-dimensional partitions, can be observed in \autoref{fig:2-d-grid}.} $(|I_1|=|I_m|=O(1/m^2))$ than the inner parts. This additional $1/m$ factor in the length compensates for the worse bound from \autoref{lem:markovbros}.

\autoref{lem:kkp-technical-lemma}, with $q$ set to $\infty$, is used to bound the error of the median-based recovery procedure, in terms of $\|p\|_{\cube_1,\infty}$. 
This allows the $\log\frac{\|p\|_{\cube_1,\infty}}{\sigma}$ iterations to be all that is further needed to bring the error down to $3\sigma$.

To avoid the run-time dependence on $\max_{x\in[-1,1]}|p(x)|/\sigma$, which maybe unknown or too large, KKP first run an $\ell_1$ regression, which gives an $\ell_\infty$ error of at most $O(d^2\sigma)$, and then run the median-based recovery algorithm on the residual polynomial, which in $\log d$ iterations drops the error further to at most  $3\sigma$.
\autoref{lem:kkp-technical-lemma}, with $q=1$, is used to bound the $\ell_1$-error of the $\ell_1$-minimizer by $O(\sigma)$. 
A further application of \autoref{lem:markovbros} bounds the $\ell_\infty$-error of the $\ell_1$-minimizer by $O(d^2\sigma)$. 
This then allows for a bound of $\log d$ on the number of iterations needed, and hence the algorithm's run-time.

\subsubsection{Our Results}
\paragraph{Generalizing to the multivariate case ($n>1$):}
We show that the idea of KKP generalizes to the multivariate setting by considering a tensorization of the Chebyshev partition, i.e., we divide the cube $[-1,1]^n$ into $m^n$ cells according to a grid partition, where each axis is divided into $m$ intervals, according to the size-$m$ Chebyshev partition of $[-1,1]$ defined by KKP. 
The analysis takes steps similar to the analysis done by KKP, and some of the proofs follow by `tensoring' KKP's arguments in some sense. 

There are some subtleties that we take care of along the way.
We successfully show optimal sample complexity results, in terms of the dependence on $d^n$, for the median-based recovery algorithm, while for the $\ell_1$ regression, we need more samples.
For this reason, we first analyze the median based recovery algorithm, the running time (but not the sample complexity) of which, depends on  $\max_{\bm x\in\cube_n}|p(\bm x)|$. 
Later, we show that by running the $\ell_1$ regression on weighted averages (with respect to cells) as the first step, we reach a constant approximation factor in bounded run-time at the cost of increasing the exponent in the sample complexity from $n$ to $O(n^2)$.

\paragraph{Overview of the algorithms and analyses:}
For using the median-recovery algorithm, we devise a multivariate analog (\autoref{thm:optimal-l-infty-bound}) of \autoref{lem:kkp-technical-lemma} for the $\ell_\infty$ norm. Specifically, we show that, for large enough $m$, every $n$-variate, individual degree-$d$ polynomial $p$ is well approximated by any piece-wise constant function with respect to the  $(m,n)$-Chebyshev partition that matches $p$ on at least one point in each cell. This is proved by a repeated application of the \emph{univariate} $\ell_\infty$ approximation statement from \autoref{lem:kkp-technical-lemma}. 
{Algorithmically}, we then do median-based recovery on a fine enough Chebyshev partition of $\cube_n$, and iteratively improve the output of the $\ell_\infty$ regression. After at most $\log (\|p\|_{\cube_n,\infty}/\eta)$ iterations, we achieve an $O(\sigma)+\eta$ approximation. A $\mathrm{poly}(\log \|p\|_{\cube,\infty},M,\log(1/\eta))$ run-time is thus achieved. One may set 
$\eta=\sigma$ to achieve an $O(\sigma)$ approximation, in this case the run-time is dependent on $\log (1/\sigma)$ instead.

We also consider the \emph{finite bit precision} setting where the samples are represented using at most  $N$ bits of precision. This forces $\sigma\geq 2^{-N}$, and the (location, evaluation) pairs of the random input samples are now rounded to $N$ bits. 
In this case, the samples' locations are not exact, and hence we are uncertain as to which Chebyshev cell they belong to. To deal with it,  we discard samples that lie in a small $\ell_1$ neighborhood of the boundary of the cells. (See \autoref{fig:refinement} for an illustration.) We then apply the median-based recovery algorithm on only the remaining samples in the cells' interior. 
The interior \emph{refined} sample points, by virtue of being far enough from their nearest cell boundary, would have remained in their respective cells, even after suffering from the rounding noise. 
Hence, we only have to ensure that all the interior regions have enough \emph{good} samples, which we show increases the sample complexity by a factor dependent only on $n$. It still gives a  tight upper bound on the sample complexity, in terms of $d^n$.
\begin{figure}[ht]
        \centering
        \includegraphics[scale=0.4]{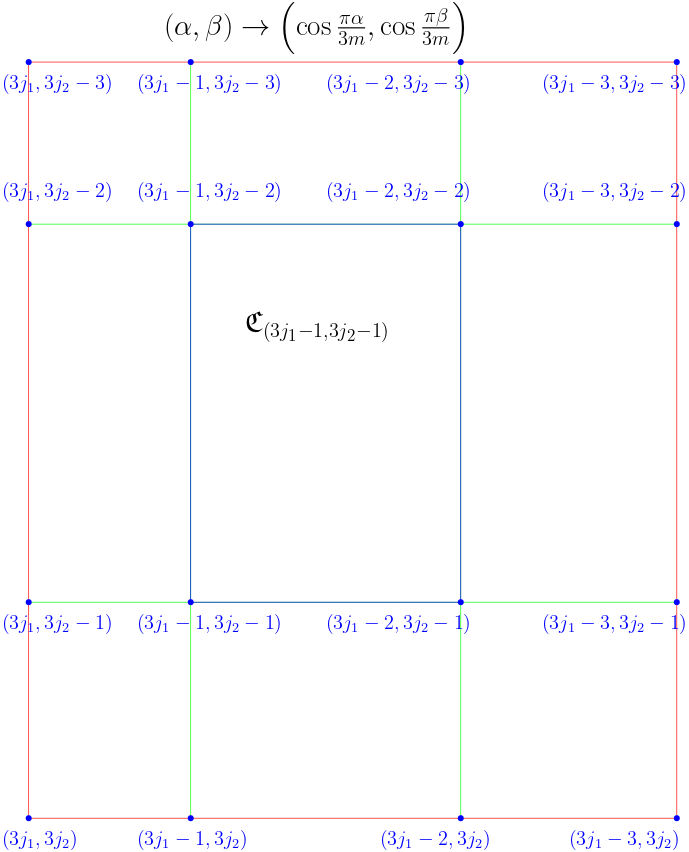}
        \caption{An illustration of cell-refinement in 2-dimensional Chebyshev grids: a $(3m,2)$-grid (in green) super-imposed on a $(m,2)$-Chebyshev cell $\cube_{(j_1,j_2)}$ (in red). The samples from middle-most cell $\cells_{(3j_1-1,3j_2-1)}$ (in blue) only are retained, and median-recovery is applied on them.}
        \label{fig:refinement}
\end{figure}

In order to avoid a run-time dependence on $(\|p\|_{\cube_n,\infty},\sigma)$, e.g., in case $\sigma$ is unknown or $\|p\|_{\cube_n,\infty}$ is too big,  we compute an $\ell_1$ minimizer $\hat{p}_{\ell_1}$ first as in KKP's approach,
However, for this analysis, we need a multivariate analog of \autoref{lem:kkp-technical-lemma} for the $\ell_1$ norm. 
The main difficulty here is the fact that now we have many more `peripheral' cells, i.e, cells on the boundary of $\cube_n$ (these cells correspond to the peripheral intervals $I_1$ and $I_m$ from the 1-dimensional Chebyshev partition, that needed Markov Brothers' Inequality).
Since these peripheral cells are narrower, and much more in number, as $n$ grows, this issue becomes more crucial. 
For example, for $n=1$,  the fraction of `peripheral' intervals is $2/m$; but for $n=2$, it is $\frac{4(m-1)}{m^2}=\frac{2}{m}(2-2/m) \gg \frac{1}{m^2}$. 
We circumvent this difficulty with our \emph{second new technical contribution} (\autoref{thm:improved-sandwich}), that relates the $\ell_\infty$ and $\ell_1$ norms of any individual degree-$d, n$-variate polynomial.
\vspace{-2ex}
\paragraph{Relating $\ell_\infty$ and $\ell_1$ norms of $p$.}
We inductively show the existence of a subset of points in $\cube_n$, with a large measure (at least $1/(2d^2)^n$), on which the valuations of $p$ can be guaranteed to be large, i.e., at least $\max_{\bfx\in\cube_n}|p(\bfx)|/2^n$. Thus we lower bound the $\ell_1$ norm of $p$ by a $1/\poly(d^n)$ factor of its $\ell_\infty$ norm, in the form of \autoref{thm:improved-sandwich}. We also note the tightness of this bound, by showing a family of polynomials for which their (resp.) $\ell_1$ norms are upper bounded by a matching $1/\poly(d^n)$ factor of the (resp.) $\ell_\infty$ norms, in the form of \autoref{lem:tightness_of_induction}.  

To begin, for any point in $\mathbf{x}\in\cube_n$, using Markov Brothers' inequality, we show the existence of a long enough line segment, on an axis-parallel line passing through it, such that $p$ on that line segment has all valuations at least $p(\mathbf{x})/2$. 
For constructing (higher) $(k+1)$-dimensional cubes from $k$-dimensional cubes, in the induction step, we prove that
all new unique line segments (one each corresponding to every point in the $k$-dimensional cube) can be translated to form a $(k+1)$-dimensional cube with a large enough Lebesgue measure. Thus, a sizable subset of points in $n$ dimensions is constructed. 
On each of these points, the valuations of $p$ are at least half of the valuations of $p$ on the corresponding points in the $k$-dimensional cube. Using the inductive hypothesis, we conclude the argument.

Using \autoref{thm:improved-sandwich} we bound the $\ell_\infty$ error of the $\ell_1$ minimizer $\hat{p}_{\ell_1}$ by $\poly(d^n)\sigma$.
We then feed $\hat{p}_{\ell_1}$ to the median based recovery procedure, which in $O(n\log d)$ iterations\footnote{The number of iterations in this case additionally depends on $O(\log (1-2\rho))$, which diverges as $\rho\rightarrow 0.5$, thus agreeing with $\rho<0.5$ being information-theoretically necessary.
}, brings down the error to $O(\sigma)$, thus proving \autoref{cor:final-cor-with-l-1-intro}.

\paragraph{Organization.}
We begin by setting up some preliminaries in \autoref{sec:prelims}. Discussion of the upper bounds follows, with \autoref{thm:l_infty_regression_upper_bound_intro} in \autoref{sec:l-infty-without-l-1}, followed by \autoref{thm:finite_presicion_intro} in \autoref{sec:finite_presicion_case}. They rely on \autoref{thm:optimal-l-infty-bound}, which is proved in \autoref{subsec:optimal-l-infty-bound}. \autoref{sec:l-infty-to-l-1} is devoted to proving $\ell_\infty$ and $\ell_1$ norms relations, in the form of \autoref{thm:improved-sandwich} and \autoref{lem:tightness_of_induction}. Among them \autoref{thm:improved-sandwich} is used to prove \autoref{cor:final-cor-with-l-1-intro} in \autoref{sec:l-infty-with-l-1}. Lower bounds are discussed in \autoref{sec:lower-bound}.

\section*{Acknowledgements}
The authors would like to thank Yuval Filmus for fruitful discussions about some aspects of the robust regression problem. 

\section{Preliminaries}\label{sec:prelims}
\paragraph{Notations.}As mentioned earlier, $\Pd$ denotes the class individual degree-$d$ polynomials, 
where a polynomial $p\colon\R^n\to\R$ is said to be of individual degree-$d$ if it can be written as
\[
p(x_1, \dots, x_n) = \sum_{\bm \alpha\in \{0,1,\dots,d\}^n} c_{\bm \alpha}x_1^{\bm \alpha_1} \cdots x_n^{\bm \alpha_n},
\]
for some set of coefficients $c_{\bm \alpha}\in\R$. We use $[m]=\{1,\hdots,m\}$, and bold font for multi-index. For example $\bmj=(j_1,\hdots,j_n)\in [m]^n$ where each entry $j_i\in [m]$. We use the math bold font $\bfx,\bfy$ for vectors. To denote random uniform sampling from a set $\caD$, we use $\sim{\caD}$.  We denote by $\cube_n=[-1,1]^n$ the $n$ dimensional solid cube and omit the subscript $n$ when it is clear from the context. Our main problem of interest is the \problemDefn, formally described in \autoref{sec:intro}.

\begin{defn}\label{defn:norms}[Norms]
    For any bounded subset $S\subsetneq \R^n$, for any $1\le q\in\R$, the $\ell_q$ norm of a function $f:\R^n\to\R$ on $S$, provided it exists, is defined as:
    \[\|f\|_{S,q}\triangleq\left(\int_S |f(\mathbf{x})|^q d\mathbf{x}\right)^{\frac{1}{q}}<\infty.\]
    The supremum norm of $f$ on $S$ is defined as $\|f\|_{S,\infty}\triangleq \lim_{q\to \infty}\|f\|_{S,q}=\sup_{\mathbf{x}\in S}\{|f(\mathbf{x})|\}$.
\end{defn}

\begin{lem}[Markov Brothers' Inequality \cite{markovbros}]\label{lem:markovbros}
    Let $p:\R\to\R$ be a degree-$d$ polynomial. Then, for all $a<b\in\R$,
    \[\|p'\|_{[a,b],\infty}\leq\frac{2d^2}{b-a}\|p\|_{[a,b],\infty}.\]
\end{lem}

\paragraph{Chebyshev Partition.} 

We partition the cube $\cube_n$ into $m^n$ cells by tensorizing the partition used by KKP for the line segment $[-1,1]$.

\begin{defn}[Chebyshev partition]\label{defn:Chebyshev-cells}
The $(m,n)$-Chebyshev partition of the cube $\cube$ is a set of $m^n$ cells indexed by $\bm j\in [m]^n$ and denoted  $\cube_{\bm j}$, such that  
$$\cube_{\bm j}=\left[\cos \frac{\pi \bm j_1}{m},\cos \frac{\pi (\bm j_1-1)}{m}\right]\times\dots\times \left[\cos \frac{\pi \bm j_n}{m},\cos \frac{\pi (\bm j_n-1)}{m}\right].$$
The grid is induced by partitioning $[-1,1]$ between the extrema points of the degree $m$ Chebyshev polynomial\footnote{see \autoref{defn:chebyshev-polynomial}, and \autoref{defn:chebyshev-extremas}} of the first kind, $T_m$, simultaneously along each axis.
\end{defn}

We generalize KKP's notion of \emph{goodness} that restricts the number of outliers in each cell:
\begin{defn}[$\alpha$-good sample set]\label{defn:alpha-good}
    We say that a set of samples $S=\{(\mathbf{x}_i,y_i)\}$ is $\alpha$-good for the $(m,n)$ Chebyshev partition,     if for every $\bmj\in[m]^n$, the fraction of outliers in the cell $\cube_{\bmj}$     is less than $\alpha$.
\end{defn}

\section{Main algorithmic result}\label{sec:l-infty-without-l-1}
In this section, we present the algorithm that solves the \problemDefn, proving the following theorem, handling an approximation factor as close to $2$ as we want, and any success probability $1-\delta$.
\begin{restatable}{thm}{linftyupperbound}[Generalized version of \autoref{thm:l_infty_regression_upper_bound_intro} ]\label{thm:l_infty_regression_upper_bound}
   Let $\eps\in(0, 1/2],\delta\in (0,\eps],\sigma\geq 0,\adderr>0$, and  $\rho<1/2$. 
   There is an algorithm that almost solves the \problemDefn 
   up to an additive error of $\adderr$. The  output of the algorithm  is   a polynomial $\hat{p}$ of degree at most $d$ in each variable, such that with probability at least $1-\delta$ (over the random input samples), $\hat{p}$ satisfies  
   \[|p(\bfx)-\hat{p}(\bfx)|\leq (2+\eps)\sigma+\eta\qquad \text{ for all } \bfx\in\cube.\]
 It uses   
 $M=O_{n,\rho}((d/\eps)^n\log(d/\delta))$ samples drawn from the multidimensional Chebyshev distribution, or $M=O_{n,\rho}((d/\eps)^{2n}\log(d/\delta))$ if the samples are drawn from the uniform measure.
 Its run-time is  that of solving $O(\log_{1/\eps}(\norm{p}_{\cube_n,\infty}/\eta)$ linear programs with $(d+1)^n<M$ variables, and $M$ constraints. 
\end{restatable}

\begin{remark}\label{rem:dep-on-rho}
   One may consider the case of non-constant values of $\rho<1/2$. Here, the number of samples increases as $\rho\to 1/2$, since the dependence of $M$ on  $\rho$ is $M\propto 1/(1-2\rho)^2$. 
\end{remark}

We remark that the idea is to show that we may achieve a multiplicative approximation factor $C$, as close to $2$ as we want (as long as $C>2$), at the cost of more samples. We may allow larger values of $\eps$, and then run our algorithm\footnote{Having $\eps'\leq 1/2$ is a limitation of the current analysis. An open question remains to make it work efficiently for any $\eps'>0$.} with $\eps'=\min\{\eps,1/2\}$. For $\eps\ge 1/2$, the dependence on $\eps$ in the sample complexity becomes constant for constant values of $n$.

\begin{remark}
In case $\sigma> 0$ is known, one may choose $\eta=\eps\sigma/2$, and set the $\eps$ parameter to be half of the desired bound to guarantee $\|\hat{p}-p\|_{\cube_n,\infty}\le (2+\eps)\sigma$. 
\end{remark}

Our  algorithm, given in \autoref{alg:median_rec} with its subroutine \autoref{alg:refinement},  is essentially the same algorithm proposed by KKP, which now uses the $(m,n)$-Chebyshev partition of the cube $\cube$ instead of the $(m,1)$-Chebyshev partition of the interval $[-1,1]$ used in KKP. Compared to their algorithm, we don't use the $\ell_1$ regression as the first step, but instead start with the $0$ polynomial as the first approximation.  
We first describe the idea of the algorithm and the median-based recovery.
\paragraph{Median-based Recovery:}
As in KKP (a similar approach was taken by  \cite{DDL2018}), for every $\bmj\in[m]^n$, we take the median $\tilde{ y}_{\bm j}$ of all the $y_i$'s corresponding to locations $\mathbf{x}_i$'s that land in the cell $\cube_{\bm j}$. We assume that the sample set $S$ is $\alpha$-good, so the fraction of outliers in each cell is strictly less than one-half (since $\alpha<1/2$) so that the median lies in between the values of the {inlier} labels. 
However, $\tilde{y}_{\bm j}$ may not itself be an inlier label for some sampled location $\bfx_i$. 
We generalize KKP's techniques to show that fitting an arbitrary $\tilde{\bfx}_{\bmj}\in \cube_{\bmj}$ to the label $\tilde{y}_{\bmj}$ yields a good algorithm. 
We  compute the polynomial $r$, that minimizes $\max_{\bm j\in[m]^n}|r(\tilde{\bfx}_{\bmj})-\tilde{y}_{\bmj}|$, and show that $r$ is $O(\sigma)$-close to $p$ in $\ell_\infty$ up to an additive error of $\eps\|p\|_{\cube_n,\infty}$. To deal with this error, we iteratively refine the estimate $r$. After $\log_{1/\eps}(\|p\|_{\cube_n,\infty}/\adderr)$ iterations, the additive error becomes as small as $\adderr$. 

\begin{algorithm}
\caption{Refinement}\label{alg:refinement}
\Procedure{\emph{\Call{Refine}{$S,\hat{p}$}}}{
  \Given {A set of samples $S=\{\mathbf{x}_i,y_i\}_{i=1}^M$, and an estimate $\hat{p}$.}
  \For{$\bm j\in[m]^n$}{
  $\Tilde{y}_{\bm j}\gets\med_{\mathbf{x}_{i}\in \cube_{\bm j}}(y_i-\hat{p}(\mathbf{x}_i))$\;
    Choose an arbitrary $\Tilde{\mathbf{x}}_{\bm j}\in \cube_{\bm j}$\;}
    Fit a degree $d$ polynomial $r$ minimizing $\|r(\Tilde{\mathbf{x}}_{\bm j})-\Tilde{y}_{\bm j}\|_{\infty}$\;
    $\hat{p}'\gets\hat{p}+r$\;
    \textbf{Return} $\hat{p}'$.
  }  
\end{algorithm}

\begin{algorithm}
\caption{Median Based Recovery}\label{alg:median_rec}
  \Given {A set of samples $S=\{\mathbf{x}_i,y_i\}_{i=1}^M$, approximation factor $\eps\le 1/2$. accuracy parameter $\adderr>0$.}
  $\hat{p}^{(1)}\gets \Call{Refine}{S,0}$; \Comment{Let $\hat{p}^{(1)}(\mathbf x)=\sum_{\bm\alpha\in\{0,1,\hdots,d\}^n} c_{\bm \alpha}\mathbf{x}^{\bm\alpha}$}\\
  Let $\valbnd$ be such that $|\hat{p}^{(1)}(\mathbf x)|\leq \valbnd$ for all $\mathbf x\in\cube$; \Comment{ Set $\valbnd\triangleq\sum_{\bm\alpha\in\{0,1,\hdots,d\}^n} |c_{\bm \alpha}|$}\\
  $N_{\ref{alg:median_rec}}\gets O\left(\log_{1/\eps}(\valbnd/\adderr)\right)$\;    \For{$i\in \{1,\dots,N_{\ref{alg:median_rec}}-1\}$}{$\hat{p}^{(i+1)}\gets\Call{Refine}{S,\hat{p}^{(i)}}$\;}
    \textbf{Return} $\hat{p}^{(N_{\ref{alg:median_rec}})}$.
 \end{algorithm}

To prove \autoref{thm:l_infty_regression_upper_bound}, we show that for $M$ as in the theorem, the set of samples is $\alpha$-good with high probability, and then we apply the following result. 
\begin{thm}[Absolute $\ell_\infty$ error bound]\label{thm:l-infty-guarantee}
    Let $c$ be some absolute constant, and let $\eps,\alpha<1/2$, $0<\adderr\le 1$, be parameters. For any  $m\geq cdn/\eps$, if the set $S=\{(\mathbf{x}_i,y_i)\}$ of $M$  samples is $\alpha$-good for the $(m,n)$-Chebyshev partition, then the median-based recovery 
    \autoref{alg:median_rec} 
    returns an individual degree-$d$ polynomial $\hat{p}=\hat{p}^{(N_{\ref{alg:median_rec}})}$, such that \[\|p-\hat{p}\|_{\cube_n,\infty}\leq(2+\epsilon)\sigma+\eta.\]
\end{thm}

The first part of the proof of \autoref{thm:l-infty-guarantee} follows the skeleton of the proof of \cite[Theorem 1.4]{kkp}, whilst skipping the $\ell_1$ intermediate regression.

The main ingredient for proving \autoref{thm:l-infty-guarantee} is the following technical result,
bounding the $\ell_\infty$ error of the non-robust $\ell_\infty$ minimizer, i.e., a single run of \autoref{alg:refinement}. This is later used to bound the error of the robust minimizer \autoref{alg:median_rec}. 
\begin{lem}[Relative $\ell_\infty$ error bound, generalization of~{\cite[Lemma 1.3]{kkp}}]\label{lem:error_bound}
    Let $c>0$ be an absolute constant. Let  $\eps,\alpha< 1/2$, and $m\ge cdn/\eps$. 
    Let the set $S=\{(\mathbf{x}_i,y_i)\}$ of $M$ samples is $\alpha$-good for the $(m,n)$-Chebyshev partition.
    And for every $\bm j\in[m]^n$, let  $\tilde{\bfx}_{\bmj}$ be an  arbitrary point from the cell $\cube_{\bmj}$, and $\tilde{y}_{\bmj}\triangleq\med_{S} \{y_i:\mathbf{x}_i\in\cube_{\bm j}\},$ i.e. the median of all those $y_i$'s in $S$, whose corresponding $\bfx_i$ is in the cell $\cube_{\bmj}$.
    Then, with
    \begin{equation}\label{eq:l_infty_minimizer}
        \hat{p}\triangleq\arg\min_{q\in\Pd}\max_{\bm j\in[m]^n}|q(\Tilde{\mathbf{x}}_{\bm j})-\Tilde{y}_{\bm j}|,
    \end{equation}
    where the minimization is over the class  $\Pd$  of all individual degree-$d$ polynomials over $\R^n$,
    we have \[\|p-\hat{p}\|_{\cube_n,\infty}\leq (2+\eps)\sigma+\eps\|p\|_{\cube_n,\infty}.\]
\end{lem}
The proof of this statement mirrors the proof of its univariate counterpart \cite[Lemma 1.3]{kkp}, with \autoref{thm:optimal-l-infty-bound} (which we prove in \autoref{subsec:optimal-l-infty-bound}) replacing \autoref{lem:kkp-technical-lemma}. Hence it is deferred to \autoref{subsec:proof-from-subsec:first-main-thm}.

\begin{proof}[Proof of \autoref{thm:l-infty-guarantee}]
    Let $\hat{p}^{(t)}$ be the individual degree-$d$ polynomial which is the $t^{th}$ estimate of $p$ computed by \autoref{alg:median_rec}, $e_t \triangleq p-\hat{p}^{(t)}$ be the $t^{th}$ error polynomial, where we define $\hat{p}^{(0)}\equiv 0$,
    and $r_t(\mathbf{x})\triangleq\arg\min_{r\in\Pd}\|r(\Tilde{\mathbf{x}}_{\bm j})-\Tilde{y}_{\bm j}\|_{\infty}$, for $\Tilde{y}_{\bm j}\triangleq\med_{S}\{y_i-\hat{p}^{(t)}(\mathbf{x}_i):\mathbf{x}_i\in\cube_{\bm j}\}$,
    such that $\hat{p}^{(t+1)}=\hat{p}^{(t)}+r_t$.
    Note that for every inlier sample $i\in[M]$, we have   $|e_t(\mathbf{x}_i)-(y_i-\hat{p}^{(t)}(\mathbf{x}_i))|=|p(\mathbf{x}_i)-y_i|\leq \sigma$, and therefore $S_t\triangleq\{(\mathbf{x}_i,y_i-\hat{p}^{(t)}(\mathbf{x}_i))\}_{i=1}^M$ is $\alpha$-good for $e_t$. So, by \autoref{lem:error_bound},
    \[\|r_t-e_t\|_{\cube_n,\infty}\leq (2+\eps)\sigma+\eps\|e_t\|_{\cube_n,\infty}.\]
    For every $t$, we have $r_t-e_t=(\hat{p}^{(t+1)}-\hat{p}^{(t)})-(p-\hat{p}^{(t)})=\hat{p}^{(t+1)}-p=-e_{t+1}$. So, we may deduce the inductive relation: $$\|e_{t+1}\|_{\cube_n,\infty}\leq (2+\eps)\sigma+\eps\|e_t\|_{\cube_n,\infty}.$$ 
    We use this relation to  bound the error at step $t$:
    \begin{align}\label{eq:error-form}
        \|p-\hat{p}^{(t)}\|_{\cube_n,\infty} &=\|e_{t}\|_{\cube_n,\infty}\nonumber\\
        & \leq(2+\eps)(1+\eps+\hdots+\eps^t)\sigma +\eps^t\|e_0\|_{\cube_n,\infty}\nonumber\\
        &\leq(2+6\eps)\sigma+\eps^t\|p\|_{\cube_n,\infty}.
    \end{align}
    The last inequality is from $\sum_{i\geq 0} \eps^i 
    \leq 1/(1-\eps)\leq 1+2\eps$ for $\eps\leq 1/2$, as well as $e_0\equiv p$.
    So, for any $t\geq \log_{1/\eps}(\|p\|_{\cube_n,\infty}/\eta)$ we have  
    \[\|p-\hat{p}^{(t)}\|_{\cube_n,\infty}\leq (2+6\eps)\sigma+\eta.\]
    Note that, after one iteration we already have $\hat{p}^{(1)}$, an approximation of $p$, which might not be the best approximation, but we can still learn some bound on $\|p\|_{\cube_n,\infty}$ from it. 
    We will show that  $N_{\ref{alg:median_rec}}=O\left(\log_{1/\eps}(\valbnd/\eta)\right)$ iterations, for $\valbnd=\sum_{\bm\alpha\in\{0,1,\hdots,d\}^n} |c_{\bm \alpha}|\geq \|\hat{p}^{(1)}\|_{\cube_n,\infty}$, where $c_{\bm \alpha}$'s are the coefficients\footnote{Given the coefficient representation of $\hat{p}^{(1)}$ from the first iteration, $\valbnd$ can be computed efficiently.} of the polynomial $\hat{p}^{(1)}$, suffice. 
   By \autoref{lem:error_bound} and the triangle inequality, after the first iteration, we have  
   \begin{align*}
   & \|\hat{p}^{(1)}-p\|_{\cube_n,\infty}\le (2+\eps)\sigma+\eps\|p\|_{\cube_n,\infty}\le (2+\eps)\sigma+\eps(\|\hat{p}^{(1)}\|_{\cube_n,\infty} +\|\hat{p}^{(1)}-p\|_{\cube_n,\infty}).
\end{align*}
Rearranging and using $\eps\leq 1/2$, we get 
\[ \|\hat{p}^{(1)}-p\|_{\cube_n,\infty}\le \frac{1}{1-\eps}((2+\eps)\sigma+\eps\|\hat{p}^{(1)}\|_{\cube_n,\infty})\leq 5\sigma+\|\hat{p}^{(1)}\|_{\cube_n,\infty}.\]
Again using the triangle inequality, we conclude: 
\begin{align*}
    \|p\|_{\cube_n,\infty}\le \|\hat{p}^{(1)}\|_{\cube_n,\infty}+\|\hat{p}^{(1)}-p\|_{\cube_n,\infty}\le 5\sigma+2\|\hat{p}^{(1)}\|_{\cube_n,\infty}.
\end{align*}
Plugging this into \eqref{eq:error-form}, we get
\begin{align*}
    \|p-\hat{p}^{(t)}\|_{\cube_n,\infty}&\leq (2+6\eps)\sigma+\eps^t(5\sigma+2\|\hat{p}^{(1)}\|_{\cube_n,\infty})\\
    &\leq(2+6\eps)\sigma+\eps^t(5\sigma+2\valbnd).
\end{align*}

Set $N_{\ref{alg:median_rec}}=\log_{1/\eps}((5+2\valbnd)/\eta)+1$. Then if $\sigma\leq 1$, we have $N_{\ref{alg:median_rec}}\geq \log_{1/\eps}(\|p\|_{\cube_n,\infty}/\eta)$, and for $t=N_{\ref{alg:median_rec}}$ we have the desired approximation, after a further rescaling of $\eps$ to $\eps/6$.
Otherwise, for $\sigma\ge 1$ we have  
\[\|p\|_{\cube_n,\infty}/\sigma\le 5+2\|\hat{p}^{(1)}\|_{\cube_n,\infty}/\sigma\leq 5+2\|\hat{p}^{(1)}\|_{\cube_n,\infty}. \]

Setting $N_{\ref{alg:median_rec}}\geq \log_{1/\eps}(\|p\|_{\cube_n,\infty}/\sigma)+1$ results in additive error at most $\eps\sigma$. 
Thus, for $t=N_{\ref{alg:median_rec}}$ we have,
\[\|p-\hat{p}^{(t)}\|_{\cube_n,\infty}\leq (2+6\eps)\sigma+\eps\sigma=(2+7\eps)\sigma.\]
     Rescaling of $\eps$ to $\eps/7$ gives the desired bound on the regression error with the additive error $= \eta>0$. 
\end{proof}

Before proving \autoref{thm:l_infty_regression_upper_bound}, we note a useful result from the analysis of multivariate functions.
\begin{thm}[Fubini-Tonelli Theorem (Theorem 14.2, \cite{Emmanuele})]\label{thm:fubini-tonelli}
    For any $X,Y\subseteq\R$, and some measurable $f:X\times Y\to\R_{\geq 0}$,
    \[\int_{X\times Y} f(x,y)d(x,y)=\int_Y\left(\int_X f(x,y)dx\right)dy=\int_X\left(\int_Y f(x,y)dy\right)dx.\]
    Furthermore, if $f(x,y)\equiv f_1(x)f_2(y)$, for some measurable $f_1:X\to\R_{\geq 0}$, and $f_2:Y\to\R_{\geq 0}$, then
    \[\int_{X\times Y} f(x,y)d(x,y)=\left(\int_X f_1(x)dx\right)\left(\int_Y f_2(y)dy\right).\]
\end{thm} 
We also note a lower bound on the width of Chebyshev grid cells.
\begin{obs}\label{obs:side-length-of-a-cell-lb}
    For a cell $\cube_{\bm j}$, denote  its  side  length on direction $t$ by $\cube_{\bm j}(t)$ then
   
    \begin{equation}
        |\cube_{\bm j}(t)|=\left|\cos\frac{\pi\bm j_t}{m}-\cos\frac{\pi(\bm j_t-1)}{m}\right|=\left|2\sin\frac{\pi}{2m}\sin\frac{\pi(2\bm j_t-1)}{2m}\right| \ge  \left(\frac{\pi}{2m}\right)^2.    \end{equation}
    The first equality follows by the trigonometric identity $\cos \theta-\cos \varphi=-2\sin((\theta+\varphi)/2)\sin((\theta-\varphi)/2)$. The next inequality holds since  
     $\sin \theta\ge \theta/2$ for all $0\le \theta\le \pi/2$ (follows from the Taylor approximation $\sin\theta \geq \theta-\theta^3/6$ for $\theta\ge 0$). For the first factor, it is used with $\theta=\frac{\pi}{2m}$ and for the second factor with $\theta=\frac{\pi}{2m}\min\{2\bm j_t-1,2( m+1-\bmj_t)-1\}$ noting that  $\sin\frac{\pi(2\bm j_t-1)}{2m}=\sin\frac{\pi(2( m+1-\bmj_t)-1)}{2m}$. 
      We also use $1\le \min\{2\bm j_t-1,2( m+1-\bmj_t)-1\}\le m.$
    
\end{obs}

\begin{proof}[Proof of \autoref{thm:l_infty_regression_upper_bound}]\label{proof:thm:l_infty_regression_upper_bound}
We need to show that the input set of samples $S$  is $\alpha$-good for the $(m,n)$-Chebyshev partition for some $\alpha<1/2$ with probability at least $1-\delta$. Then for $m=cdn/\eps$, for some constant $c>0$, the result follows from \autoref{thm:l-infty-guarantee}.
Consider the $(m,n)$-Chebyshev partition. For any $\bmj\in [m]^n$, let $p_{\bm j}$ be the probability that a sample from $\chi$ is in $\cube_{\bm j}$.
Let $X_{\bm j}$ be the number of samples in $\cube_{\bm j}$, and $Y_{\bm j}$ be the number of outliers in $\cube_{\bm j}$. 

The probability that $S$ violates the  $\alpha$-goodness for a cell $\cube_{\bm j}$ can be bounded  by the law of total probability,
\begin{align*}
    \Pr[Y_{\bm j}\geq\alpha X_{\bm j}]&\leq\Pr[Y_{\bm j}\geq\alpha X_{\bm j}| X_{\bm j}\leq Mp_{\bm j}/2]\underbrace{\Pr[X_{\bm j}\leq Mp_{\bm j}/2]}_{(I)}\\
    &+\underbrace{\Pr[Y_{\bm j}\geq\alpha X_{\bm j}| X_{\bm j}>Mp_{\bm j}/2]}_{(II)} \Pr[X_{\bm j}>Mp_{\bm j}/2].
\end{align*}
Next, we bound each term separately. For the first term, note that $\mathbb{E}[X_{\bm j}]=Mp_{\bm j}$, so using Chernoff bound we have, 
\[(I)\leq\Pr[|X_{\bm j}-\mathbb{E}[X_{\bm j}]|\geq Mp_{\bm j}/2]\leq 2e^{-Mp_{\bm j}/12}.\]
For the second term, $\mathbb{E}[Y_{\bm j} | X_{\bm j} = x]=\rho x$, and using Hoeffding's Inequality, we have
\[\Pr[Y_{\bm j}\geq\alpha X_{\bm j}\mid X_{\bm j}]\leq\Pr[|Y_{\bm j}-\mathbb{E}[Y_{\bm j}]|\geq(\alpha-\rho)X_{\bm j}\mid X_{\bm j}]\leq 2e^{-(\alpha-\rho)^2 X_{\bm j}}.\]
Setting $\alpha=\frac{2\rho+1}{4}$ we have $\alpha<\frac{1}{2}$, and $(\alpha-\rho)^2=\left(\frac{1-2\rho}{4}\right)^2=(1-2\rho)^2/16$, giving us   
\[(II)\leq 2e^{-(1-2\rho)^2X_\bmj/16}\le 2e^{-(1-2\rho)^2Mp_{\bm j}/32}.\]
The last inequality is by the condition $X_\bmj>Mp_\bmj/2$.
 We conclude that the failure probability is 
 \begin{equation}
     \label{eq:alpha-good-prob}
     \Pr[\exists\bmj\in[m]^n:Y_{\bmj}\geq\alpha X_{\bmj}]\leq \sum_{\bmj\in[m]^n}\Pr[Y_{\bmj}\geq\alpha X_{\bmj}]\leq\sum_{\bmj\in[m]^n}\left( 2e^{-(1-2\rho)^2Mp_{\bm j}/32}+2e^{-Mp_{\bm j}/12}\right).
 \end{equation}
 
The first inequality is by a union bound over all the $m^n$ cells. For the second inequality, we plugged in the bounds for $(I),(II)$. 
       By \autoref{obs:side-length-of-a-cell-lb}, we deduce that for any $\bm j\in[m]^n$, a point uniformly sampled from $\cube$ falls into $\cube_{\bm j}$ with probability 
       \[p_{U}(\cube_{\bm j})=\frac{V_n(\cube_{\bm j})}{2^n}\ge\left(\frac{\pi^2}{8m^2}\right)^n\ge m^{-2n}.\] 
       For sampling from the unidimensional Chebyshev distribution, using the fact that $\int\frac{dx}{\sqrt{1-x^2}}=\arcsin(x)$, KKP observed: for any $j\in[m]$,
       \begin{align}\label{eq:unidim-chebyshev-measure}
           \int_{\cos(\pi j)/m}^{\cos(\pi(j-1)/m)}\frac{1}{\sqrt{1-x^2}}dx= \arcsin{\left(\cos\frac{\pi(j-1)}{m}\right)}-\arcsin{\left(\cos\frac{\pi j}{m}\right)}=\frac{\pi}{m}.
       \end{align}
       So, for sampling from the $n$-dimensional Chebyshev distribution, the probability that $\bfx\in \cube_\bmj$ becomes
    \begin{align*}
        p_{C}(\cube_{\bm j})&=\int_{\cube_{\bm j}(n)}\hdots\int_{\cube_{\bm j}(1)} \frac{1}{\pi\sqrt{1-\mathbf{x}_1^2}}\times\hdots\times \frac{1}{\pi\sqrt{1-\mathbf{x}_n^2}}\; d\mathbf{x}_1\hdots d\mathbf{x}_n \tag{Integrands are all non-negative }\\
        &=\prod_{i=1}^n\left(\int_{\cube_{\bm j}(i)}\frac{1}{\pi\sqrt{1-\mathbf{x}_i^2}}d\mathbf{x}_i\right) \tag{Splitting independent integrals by \autoref{thm:fubini-tonelli}}\\
        &=\prod_{i=1}^n\left(\frac{1}{\pi}\int_{\cos(\pi\bm j_i/m)}^{\cos(\pi(\bm j_i-1)/m)}\frac{1}{\sqrt{1-\mathbf{x}_i^2}}d\mathbf{x}_i\right) \tag{$\because\cube_{\bm j}(i)=[\cos(\pi\bm j_i/m),\cos(\pi(\bm j_i-1)/m)]$}\\
        &=\prod_{i=1}^n\frac{1}{m}=\frac{1}{m^n}. \tag{By using \eqref{eq:unidim-chebyshev-measure} and since $\bm j_i\in[m],\forall i\in[n]$}
                                                    \end{align*}
For Chebyshev sampling, with $p_{\bm j}=m^{-n}$, and upper bounding the failure probability from  \eqref{eq:alpha-good-prob} by at most $\delta$, we get $S$ is $\alpha$-good for 
\begin{equation}\label{eq:M-cheb-func-of-m}
    M_C(m)=(1-2\rho)^{-2}m^n\log(m^n/\delta).
\end{equation} 
For \autoref{alg:median_rec}, with $m=cdn/\eps$, for some constant $c>0$, this gives us a Chebyshev sample complexity of 
$M_C=\frac{1}{(1-2\rho)^2}(cnd/\eps)^n\log\frac{d}{\eps\delta}$. 
For uniform sampling, replacing $p_{\bm j}$ by the bound on  $p_U(\cube_{\bm j})\ge m^{-2n}$ and  then  plugging in 
\begin{equation}\label{eq:M-uni-func-of-m}
    M=M_U(m)=(1-2\rho)^{-2} m^{2n} \log (4m^n/\delta) 
\end{equation}
 in \eqref{eq:alpha-good-prob}, we get that  $S$ is $\alpha$-good with probability:
\begin{align*}
    \Pr[Y_{{\bm j}}<\alpha X_{{\bm j}}, \forall {{\bm j}}]&\ge 1-4\sum_{\bm j\in[m]}\frac{\delta}{4m^n}=1-\delta
\end{align*}
Thus, with $m=cdn/\eps$, the Uniform sample complexity is 
$M_U=\frac{1}{(1-2\rho)^2}(cnd/\eps)^{2n}\log\frac{d}{\delta}$, for some constant $c>0$. 
The run-time of \autoref{alg:median_rec} is thus that of solving $N_{\ref{alg:median_rec}}=O(\log_{1/\eps}( \|p\|_{\cube_n,\infty}/\adderr))$ linear programs with $O(d^n)$ variables, and $M$ constraints. 
\end{proof}

\subsection{Approximating polynomials by piece-wise constant functions}\label{subsec:optimal-l-infty-bound}

In this section, we prove \autoref{thm:optimal-l-infty-bound}, by generalizing \autoref{lem:kkp-technical-lemma} for the case of $\ell_{\infty}$, to the multivariate setting.
The idea here is to approximate $p$ on a fine enough Chebyshev grid, by an arbitrary piece-wise constant function $r$, that is (i) constant on every Chebyshev cell, and (ii) consistent with $p$ on some point in every cell. We show that the finer the grid is, the smaller the difference $p-r$ becomes, and hence, the better an approximation $r$ may be for $p$. We first restate \autoref{lem:kkp-technical-lemma} as our generalization is through reducing to the univariate case.

\kkptechnicallemma*
We prove how we can conclude from \autoref{lem:kkp-technical-lemma}, a similar result for the multivariate case, where we replace the interval $[-1,1]$ with the $n$ dimensional solid cube $\cube=[-1,1]^n$. The polynomial $p\colon \cube\to\R$ is now of individual degree at most $d$ (in each variable), and the function $r$ is now piece-wise constant with respect to the $(m,n)$-Chebyshev partition.

\optimallinftybound*

Before we prove \autoref{thm:optimal-l-infty-bound}, we note and prove a useful technical claim. Though its proof is simple, it exposes the idea behind \autoref{lem:kkp-technical-lemma} (at least for the $\ell_{\infty}$ case).  
 \begin{claim}\label{claim:p.w.-constant-func-exists}
    Let $g:[-1,1]\to\R$ be a degree-$d$ univariate polynomial. Consider the $m$-size Chebyshev partition of  $[-1,1]$, denoted by $I_1,\dots, I_m$. 
    For any $i\in[m]$, and any two points $\alpha,\beta\in I_i$,  there exists a function $v\colon[-1,1]\to\R$ that is piece-wise constant with respect to $g$, and the partition, such that:
    \begin{equation}\label{eq:diff-bounded-by-l-infty}
        |g(\alpha)-g(\beta)|\leq \|g-v\|_{[-1,1],\infty}.
    \end{equation}
    Furthermore, there exists a universal constant $C$ such that $|g(\alpha)-g(\beta)|\leq \frac{Cd}{m}\|g\|_{[-1,1],\infty}$
    \end{claim}
    \begin{proof}
        Fix  $i\in[m]$, and $\alpha,\beta\in I_i$. We define the piece-wise constant function  $v$ as follows, 
        \begin{itemize}
            \item For any $j\in [m]$  such that $j\neq i$, choose an arbitrary point $x_j\in I_j$, and set $v(x)=g(x_j)$ for all $x\in I_j$.
            \item For  $I_i$, set $v(x)=g(\beta)$ for all $x\in I_i$.
        \end{itemize}
           Note that,  in particular, $v(\alpha)=g(\beta)$. Hence, 
    \[|g(\alpha)-g(\beta)|=|g(\alpha)-v(\alpha)|\leq \max_{\gamma\in[-1,1]}|g(\gamma)-v(\gamma)|=\|g-v\|_{[-1,1],\infty}.\]
    By \autoref{lem:kkp-technical-lemma}, we conclude $|g(\alpha)-g(\beta)|\leq \frac{Cd}{m}\|g\|_{[-1,1],\infty}$, for some universal constant $C$.
    \end{proof}
    We are now ready to prove \autoref{thm:optimal-l-infty-bound}:
\begin{proof}[Proof of \autoref{thm:optimal-l-infty-bound}]

    Observe,
 \begin{align}\label{eq:bound-by-axis-parallel-summands}
        \|p&-r\|_{\cube_n,\infty}\\
        &=\max_{\bm j\in[m]^n,\mathbf{x}\in \cube_{\bm j}}|p(\mathbf{x})-r(\mathbf{x})|\nonumber\\
        &=\max_{\bm j\in[m]^n,\mathbf{x}\in \cube_{\bm j}}|p(\mathbf{x})-p(\mathbf{x}^{(\bm j)})|\nonumber\\
        &\leq \max_{\bm j\in [m]^n, \mathbf{x},\bm \bfy\in \cube_{\bm j}}|p(\mathbf{x})-p(\bfy)|\\
        &\leq \max_{\bm j\in [m]^n, \mathbf{x},\bfy\in \cube_{\bm j}} \sum_{k=1}^n |p(x_1,\hdots,x_{k-1},x_k,y_{k+1},\hdots,y_n) - p(x_1,\hdots,x_{k-1},y_k,y_{k+1},\hdots,y_n)|,\nonumber
    \end{align}
where the last inequality is by a hybrid argument: walking from the point $\mathbf{x}$ to the point $\bfy$ along the axes, gives the successive summands in \eqref{eq:bound-by-axis-parallel-summands}. 
 We next bound each of the summands by $C\frac{d}{m}\|p\|_{\cube_n,\infty}$, then summing up all $n$ of them results with the desired bound on $\|p-r\|_{\cube_n,\infty}$. 
     For every $k\in[n]$,
            we observe that \[(x_1,\hdots,x_{k-1},x_k,y_{k+1},\hdots,y_n) \text{, and } (x_1,\hdots,x_{k-1},y_k,y_{k+1},\hdots,y_n)\] are points on the line $\pline_{(x_1,\hdots,x_{k-1},y_k,y_{k+1},\hdots,y_n),\bm e_k}$. \\
    So, they provide evaluations of $p_{\pline_{(x_1,\hdots,x_{k-1},y_k,y_{k+1},\hdots,y_n),\bm e_k}}(t)$, the univariate line restriction polynomial, which has degree\footnote{The line $\pline_{(x_1,\hdots,x_{k-1},y_k,y_{k+1},\hdots,y_n),\bm e_k}$ being \emph{axis-parallel} is crucial here, allowing us to bound the degree of the line restriction polynomial $p_{\pline_{(x_1,\hdots,x_{k-1},y_k,y_{k+1},\hdots,y_n),\bm e_k}}(t)$ by the individual degree of $p$, as $t$ varies only along $\bm e_k$.} at most $d$, since $p:\cube\to\R$ has individual degree at most $d$, for every variable. And, since $\bfx,\bfy\in\cube_{\bm j}$, we have $x_k,y_k\in I_i$, for  $i=\bm j_k\in[m]$. Thus, by \autoref{claim:p.w.-constant-func-exists},     \begin{align*}\label{eq:bound-0n-summands}
        &|p(x_1,\hdots,x_{k-1},x_k,y_{k+1},\hdots,y_n) - p(x_1,\hdots,x_{k-1},y_k,y_{k+1},\hdots,y_n)|\nonumber\\
                \leq &\;\frac{Cd}{m}\|p_{\pline_{(x_1,\hdots,x_{k-1},y_k,y_{k+1},\hdots,y_n),\bm e_k}}\|_{[-1,1],\infty} \nonumber        \\
        \leq &\;\frac{Cd}{m}\|p\|_{\cube_n,\infty}, 
    \end{align*}
                            as needed.
\end{proof}

\section{Dealing with finite precision representations}\label{sec:finite_presicion_case}
We prove a more precise statement of \autoref{thm:finite_presicion_intro}, giving an algorithm for handling an approximation factor close enough to $2$, and for any success probability $1-\delta$. 
\begin{restatable}{thm}{bitcomplexityform}[Generalized version of \autoref{thm:finite_presicion_intro}]\label{thm:finite_presicion}
   Let $N$ be the number of bits of precision, $\sigma\ge 2^{-N}$ and, constant $\rho<1/2$.
   For any $\eps\leq 1/2$ such that $\eps=\Omega_n(d2^{-N/2})$, and $\delta\in(0,\eps]$,
   there exists an algorithm for the \problemDefn.
    The  output of the algorithm is $\hat{p}\colon\R^n\to \R$, a polynomial  of degree at most $d$ in each variable, that satisfies  
   \[|p(\bfx)-\hat{p}(\bfx)|\leq (2+\eps)\sigma \qquad\text{ for all } \bfx\in\cube_n,\]
with probability at least $1-\delta$.
    It uses   
 $M=O_{n,\rho}((d/\eps)^n\log(d/\delta))$ samples drawn from the multidimensional Chebyshev distribution, or $M=O_{n,\rho}((d/\eps)^{2n}\log(d/\delta))$ if the samples are drawn from the uniform measure.
 Its run-time is that of solving $O(\log_{1/\eps}\norm{p}_{\cube_n,\infty}+N)$ linear programs with $(d+1)^n<M$ variables, and $M$ constraints.
 \end{restatable}

 \begin{remark}
     We note that the condition on $\eps$ implies that in order to learn degree $d$ polynomials using \autoref{alg:median_rec}, one needs at least  $N=\Omega(\log d)$ bits of precision. 
     The reason is that to get a good approximation we need to take a fine enough grid, but the grid's ``fineness parameter'' $m$ is limited as well by the precision restriction, as we need the width of any cell to be at least $2^{-N}$. 
     Note that if $N=o(\log d)$, i.e. $d=2^{\Omega(N)}$, then $\eps=\Omega(1)$, i.e. the approximation factor achieved in this setting is too large, compared to the factor of at most $3$, achievable when $N=\Omega(\log d)$.
 \end{remark}

 \begin{proof}[Proof of \autoref{thm:finite_presicion}]
The algorithm that solves this problem, uses \autoref{alg:median_rec} as a black box, with $\adderr=\eps 2^{-N}$.
We note that if the samples $\bfx_i\in\cube$ are exact, i.e. described using infinite precision,
then we get the result by applying \autoref{thm:l_infty_regression_upper_bound}, and rescaling $\eps$ to $\eps/2$. 

Otherwise, we note that the only information used by \autoref{alg:median_rec} from the sample set, for every $i\in[M]$ is: (i) the value $y_i$, and (ii) the index $\bm k_i\in[m]^n$ of the cell $\cube_{\bm k_i}$ into which $\bfx_i$ lands. 
The inaccuracy\footnote{from finite-bit precision representations of reals} of (i) is covered by having $\sigma\geq 2^{-N}$. For (ii), we add a preliminary step to our algorithm, first sifting out the samples, that we don't know as to which cell $\cube_\bmj$ they belong. These are the samples which are close to the borders of the cell, i.e. within a distance of $2^{-N}$ from any border of the cell. More precisely, 
for an input sample set $S=\{(\bfx_i,y_i)\}$ of size $|S|=M$, 
we run \autoref{alg:median_rec} on the \emph{restricted} sample set $S'\subseteq S$ defined as follows:
\[S'=\left\{(\bfx,y)\in S: \forall i\in[n].\ |k_{x_i}-x_i|> 2^{-N} \right\},\]
where $k_{x_i}$ is the closest Chebyshev extrema to $x_i$.
We note the samples we omit may come, originally, from another cell or from the same cell within an additional distance of $2^{-N}$; the latter ones are good for us, but we have no way to distinguish between these two types. So, we omit them all.

We next show that even if we omit those samples, we only need to multiply the number of samples by a factor dependent only on $n$, to ensure we have enough samples for which we are sure as to which cell they belong.
The sifting process is promised to save all the samples that, originally (i.e. in their exact infinite-bit precision representation), before the noise is applied, lie in the interior of the cell at a distance of at least $2\cdot 2^{-N}$ (twice the precision noise) from any Chebyshev extrema (which determine the nearest cell boundary). 
This process requires the interior of the cell, when we omit the width $\sigma'=2\cdot 2^{-N}$ boundary, to exist. This restricts the partition to be coarse enough, i.e., we won't be able to take $m$ to be too large, which means that $\eps$ cannot be too small.  
In particular, we need the side length of a cell $|\cube_\bmj(i)|> 2\sigma'$.

        Consider the interior of a cell $\cube_{\bm j}$, denoted by $\cube_{\bm j}^\ast$, defined as the region of $\cube_{\bm j}$, within a boundary of width 
    $\sigma'\triangleq\frac{1}{4}\left(\frac{\pi}{2m}\right)^2=\left(\frac{\pi}{4m}\right)^2\geq 2\cdot 2^{-N}$, i.e. $|\cube_{\bm j}^\ast(i)|=|\cube_{\bm j}(i)|-2\sigma'$, for every $i\in[n]$. 
    By \autoref{obs:side-length-of-a-cell-lb}, $|\cube_{\bm j}(i)|\geq\left(\frac{\pi}{2m}\right)^2$, and hence, for every $ i\in[n]$ we have,
    \[|\cube_{\bm j}^\ast(i)|\geq\left(\frac{\pi}{2m}\right)^2-2\sigma'\geq\frac{1}{2}\left(\frac{\pi}{2m}\right)^2=\left(\frac{\pi}{8m}\right)^2
        .\]
    So, for uniform sampling, we get 
    \[p_U(\cube_{\bm j}^\ast)=\frac{V_n(\cube_{\bm j}^\ast)}{2^n}\geq\left(\frac{\pi^2}{16m^2}\right)^n\ge\left(\sqrt{2}m\right)^{-2n}.\]
    Plugging this into \eqref{eq:alpha-good-prob}, with $p_{\bm j}$ replaced by the bound on $p_U(\cube_{\bm j}^\ast)$, we get $S$ is $\alpha$-good for $M_U=M_U(\sqrt{2}m)$ samples, where $M_U(\cdot)$ is the formulation in \eqref{eq:M-uni-func-of-m}. Then, with $m=c_1 nd/\eps$ for some constant $c_1>0$, we get
    \begin{align*}
        M_U=M_U(\sqrt{2}m)&=(1-2\rho)^{-2}2^n m^{2n} \log (4(\sqrt{2}m)^n/\delta)\\
        &=\frac{1}{(1-2\rho)^2}\left(\frac{cnd}{\eps}\right)^{2n}\log\frac{d}{\delta},
    \end{align*}
    for some constant $c>0$. Note, this is asymptotically the same as the $M_U$ obtained for \autoref{thm:l_infty_regression_upper_bound}.
    
    For Chebyshev sampling, we consider the $(3m,n)$-Chebyshev grid $\cells\triangleq\{\cells_{\bm j},\bm j\in [3m]^n\}$, where
    \[\cells_{\bm j}=\left[\cos \frac{\pi \bm j_1}{3m},\cos \frac{\pi (\bm j_1-1)}{3m}\right]\times\dots\times \left[\cos \frac{\pi \bm j_n}{3m},\cos \frac{\pi (\bm j_n-1)}{3m}\right].\]
                            We may observe, $\cells$ is a cell-refinement of the $(m,n)$-Chebyshev grid $\{\cube_{\bm j},\bm j\in[m]^n\}$. Formally:
    \begin{claim}\label{claim:cell-refinement}
        For every $\bm j=(j_1,\hdots,j_n)\in[m]^n$, the cell $\cube_{\bm j}$ contains all the $3^n$ refined cells $\cells_{(k_1,\hdots,k_n)}$, where  $k_i\in\{3j_i,3j_i-1,3j_i-2\}$, for all $i\in[n]$, and nothing else.
    \end{claim}
    \begin{proof}
        Consider some fixed $\bm j=(j_1,\hdots,j_n)\in[m]^n$, and the cell $\cube_{\bm j}$. For every $i\in[n]$, its $i^{th}$ side $\cube_{\bm j}(i)$ is precisely the union of the $i^{th}$ sides of all the $3^n$ refined cells contained in $\cube_{\bm j}$:
    \begin{align*}
        \cube_{\bm j}(i)=&\left[\cos\left(\frac{\pi\bm j_i}{m}\right),\cos\left(\frac{\pi(\bm j_i-1)}{m}\right)\right]\\
        =&\left[\cos\left(\frac{\pi(3\bm j_i)}{3m}\right),\cos\left(\frac{\pi(3\bm j_i-1)}{3m}\right)\right] \cup \left[\cos\left(\frac{\pi(3\bm j_i-1)}{3m}\right),\cos\left(\frac{\pi(3\bm j_i-2)}{3m}\right)\right]\\
        &\cup \left[\cos\left(\frac{\pi(3\bm j_i-2)}{3m}\right),\cos\left(\frac{\pi(3\bm j_i-3)}{3m}\right)\right]
                        \tag{By continuity, and monotonicity of $\cos$ on $\cube_{\bm j}(i)$}\\
        =&\bigcup_{t=0}^{2} \left[\cos\left(\frac{\pi(3\bm j_i-t)}{3m}\right),\cos\left(\frac{\pi(3\bm j_i-(t+1))}{3m}\right)\right] =\bigcup_{t=0}^{2}\cells_{(\circledast,\hdots,\circledast,3\bm j_i-t,\circledast,\hdots, \circledast)}(i), 
            \end{align*}
    where each $\circledast_r$ (denoting $\circledast$ in the $r^{th}$ position, for all $r\in[3m]\setminus\{i\}$), could be any of the  $k\in[3m]$. The $n$ relations, one for each side, then fix each $\circledast_r$ to be one of $\{3j_i,3j_i-1,3j_i-2\}$.
    \end{proof}
    We now want to ensure that, in every cell $\cube_{\bm j}$, for every sample in the middle-most cell $\cells_{(3j_1-1,\hdots,3j_n-1)}$ the precision error doesn't move the sample to some other $\cube_{\bm j'}$. 
     So, we need 
    \[\sigma'\geq\max_{j_i\in[m]}\min\left\{\cos\frac{3j_i-1}{3m}-\cos\frac{3j_i-2}{3m},\cos\frac{3j_i-2}{3m}-\cos\frac{3j_i}{3m}\right\}\ge\left(\frac{ \pi}{6m}\right)^2.\]
    With this, to then ensure that $\cells_{(3j_1-1,\hdots,3j_n-1)}$ has sufficiently many good samples, we observe
    \begin{align*}
        p_C(\cells_{(3j_1-1,\hdots,3j_n-1)})=\frac{1}{(3m)^n}=\frac{1}{3^n}p_C(\cube_{\bm j}).
    \end{align*}
    Again, plugging this into \eqref{eq:alpha-good-prob}, with $p_{\bm j}$ replaced by $p_C(\cube_{\bm j}^\ast)$, we get $S$ is $\alpha$-good for $M_C=M_C(3m)$ samples, where $M_C(\cdot)$ is the formulation in \eqref{eq:M-cheb-func-of-m}. With $m=c_1 nd/\eps$ for some constant $c_1>0$, we get
    \begin{align*}
        M_C=M_C(3m)&=(1-2\rho)^{-2}(3m)^n\log((3m)^n/\delta)\\
        &=\frac{1}{(1-2\rho)^2}\left(\frac{cnd}{\eps}\right)^{n}\log\frac{d}{\delta},
    \end{align*}
    for some constant $c>0$. Note, this too is asymptotically the same as the $M_C$ obtained for \autoref{thm:l_infty_regression_upper_bound}.

    So, we can handle  $\sigma=2^{-N}\le \left(\frac{\pi}{6m}\right)^2$. 
    We need $m\ge c\max\{ 2^{N/2},\frac{dn}{\eps}\}$.
    Let $\sigma=2^{-N}$ for some fixed $N$, then $\sigma'\ge \sigma$ implies $m\le (\pi/6)2^{N/2}$.
    Simultaneously, we need $m>cdn/\eps$.
    This induce the condition on $\eps\ge\frac{cdn}{m}\geq c_0dn2^{-N/2}$
    (for $c_0=6c/\pi$).
\end{proof}

\section{From \texorpdfstring{$\ell_{\infty}$}{Lmax} error to \texorpdfstring{$\ell_1$}{L1} error}\label{sec:l-infty-to-l-1}

In this independent section, we prove our second technical new idea, restated here for convenience:

\linftytolone*
 This, in \autoref{sec:l-infty-with-l-1}, allows us to design and analyze a variant of \autoref{alg:median_rec}, that has a run-time independent of $\|p\|_{\cube_n,\infty}$.
To prove \autoref{thm:improved-sandwich}, we need the following result about univariate polynomials, which we next generalize for the multivariate case.

\begin{lem}\label{lem:local-sandwich-univariate} Let $g:\R\to\R$ be a degree at most $d$ polynomial. Let  $x^*\in[-1,1]$ be such that $|g(x^*)|=\|g\|_{[-1,1],\infty}$. Consider the set $S\triangleq\{x\in[-1,1]:|x^*-x|\leq 1/2d^2\}$. Then, for every $ x\in S$,$$|g(x)|\geq |g(x^*)|/2.$$
\end{lem}
To prove it, we need the following fundamental result from the analysis of univariate functions.
\begin{thm}[Mean Value Theorem (Theorem 5.10, \cite{rudin})]\label{thm:mean-value}
    Let $a<b\in\R$, and $f:[a,b]\to\R$ be continuous on $[a,b]$, and differentiable on $(a,b)$. Then there exists a point $z\in(a,b)$ such that \[f(b)-f(a)=(b-a)f'(z).\]
\end{thm}
\begin{proof}[Proof of \autoref{lem:local-sandwich-univariate}]
        Fix some $x\in S$. By \autoref{thm:mean-value}, there exists a $z$ between $x^*$ and $x$, such that $g(x^*)-g(x)=g'(z)(x^*-x)$. Since $z\in[-1,1]$, by \autoref{lem:markovbros}, we have $|g'(z)|\leq d^2\|g\|_{[-1,1],\infty}$. So,
        \begin{align*}
            |g(x^*)-g(x)|&\leq d^2\|g\|_{[-1,1],\infty}|x^*-x|\leq\frac{\|g\|_{[-1,1],\infty}}{2}=\frac{|g( x^*)|}{2}.
        \end{align*}
        The second inequality follows since $|x^*-x|\leq 1/2d^2$, for all $x\in S$.
        By the triangle inequality, we conclude
        \[|g(x)|\geq |g(x^*)|-|g(x^*)-g(x)|\geq |g( x^*)|/2.\qedhere\]
\end{proof}
    Using this, we can prove a local Lipschitz-like argument along \emph{axis-parallel} lines, for any point in $\cube$:
\begin{cor}\label{cor:local-sandwich-improved}
Let $p\colon\cube\to\R$ be a polynomial of individual degree at most $d$. For every  $\bm a\in\cube$, and $i\in [n]$, consider the \textbf{axis-parallel} line $\pline_{\bm a,\bm e_i}\triangleq\{\bm a+t\bm e_i:t\in\R\}$, passing through $\bm a$. \\
There exists a line segment $J\subset \pline_{\bm a,e_i}\cap \cube$ of length at least  $\frac{1}{2d^2}$, such that for every $\mathbf{x}\in J$,  $|p(\mathbf{x})|\geq\frac{|p(\bm a)|}{2}$.
\end{cor}
\begin{proof}
    Fix $\bm a\in\cube$ and $i\in[n]$. 
    Observe, $p$ restricted to $\pline_{\bm a,\bm e_i}$, 
    $p_{\pline_{\bm a,\bm e_i}}(t)\triangleq p(\bm a + t\bm e_i)$, i.e. we fix $x_j=\bm a_j$ for all $j\neq i$, and only let $x_i$ vary. It is a univariate polynomial in the formal variable $t$, of degree\footnote{Here $\pline_{\bm a,\bm e_i}$ being \emph{axis-parallel} is crucial, as $t$ then varies only along $\bm e_i$, ensuring the degree of $p_{\pline_{\bm a,\bm e_i}}(t)$ is bounded by the individual degree of $p$.} at most $d$. 
    Consider a point $\mathbf{x}^*\in\pline_{\bm a,\bm e_i}\cap\cube$ such that $|p(\mathbf{x}^*)|=\|p_{\pline_{\bm a,\bm e_i}}\|_{\pline_{\bm a,e_i}\cap \cube_n,\infty}$. 
    For $J\triangleq\{\bm  y\in\pline_{\bm a,\bm e_i}\cap \cube:\|\mathbf{x}^*-\bm y\|_2\leq 1/2d^2\}$, the length of $J$ is at least $\frac{1}{2d^2}$, and hence by \autoref{lem:local-sandwich-univariate}, for every $\bm y\in J$ we have: $|p(\bm y)|\geq|p(\mathbf{x}^*)|/2\geq|p(\bm a)|/2$. 
\end{proof}

Next, we wish to find a relatively large subcube of $\cube$ of side length $\frac{1}{2d^2}$ for which all the valuations of $p$ are large. Such a subcube may not exist. So instead, we prove the existence of a set of points in $\cube$, of the same measure on which $p$ can be guaranteed to have large values. We next define a measure of a set of points in a way that allows us to apply our inductive argument. This measure has the feature that when applied to the entire cube it becomes the standard integral.  

\begin{defn}[Measure]\label{defn:measure}
    For any $k\in[n]$, and a fixed \emph{affinity} point $\bm a=(a_{k+1},\hdots,a_n)\in[-1,1]^{n-k}$, consider the $k$-dimensional \emph{affine} cube $\cube(\bm a)\triangleq [-1,1]^k\times \{a_{k+1}\}\times\dots\times \{a_n\}\subseteq\cube$. 
Let $J_k\subseteq \cube(\bm a)$.     It's $k$-dimensional measure is defined as:
    \[|J_k|\triangleq\int_{\mathbf{x}\in J_k}dx_k\hdots dx_1.\]
\end{defn}

\begin{lem}[Generalization of \autoref{lem:local-sandwich-univariate}]\label{lem:gen-prod-set}
    Let $p\colon\cube\to\R$ be a polynomial of individual degree at most $d$. For any  $\bm y\in\cube$  there exists a set of points $J\subset\cube$, such that for every $\mathbf{x}\in J,|p(\mathbf{x})|\geq\frac{|p(\bm y)|}{2^n}$, and 
    \[|J| =\int_{\mathbf{x}=(x_1,\hdots,x_n)\in J} dx_n\hdots dx_1 \geq \frac{1}{(2d^2)^n}.\]
\end{lem}
\begin{proof}
    Fix some $\bm y=(y_1,\hdots,y_n)\in\cube$. Let $r\triangleq\frac{1}{2d^2}$. We build this set $J\subset[-1,1]^n$ inductively:
    
    \paragraph{Base Case ($n=1$):} Consider the line $$\pline_{\bm y,\bm e_1}\triangleq\{\bm y+t\bm e_1:t\in\R\}.$$ 
    $p$ restricted to $\pline_{\bm y,\bm e_1}$, i.e.,  $p_{\pline_{\bm y,\bm e_1}}(t)\triangleq p(\bm y + t\bm e_1)$ is a univariate polynomial in the variable $t$, of degree at most $d$.
    So, by \autoref{cor:local-sandwich-improved}, there exists a line segment  $ J_1\subset\pline_{\bm y,\bm e_1}\cap\cube$ such that $|p(\mathbf{x} )|\geq\frac{|p(\bm y)|}{2}$ for every $\mathbf{x}=( x_1,y_2,\hdots, y_n)\in J_1$, and $|J_1|=\int_{x_1\in J_1}dx_1\geq r$.
\paragraph{Induction Step $(n=k+1)$:}
    Assume  by \emph{ IH for $n=k$} that there exists $J_k\subset[-1,1]^k$ such that, for every $\mathbf{x}=(x_1,\hdots,x_k, y_{k+1},\hdots, y_n)\in J_k$, we have $|p(x_1,\hdots,x_k, y_{k+1},\hdots, y_n)|\geq\frac{|p(\bm y)|}{2^k}$, and
    \[|J_k|=\int_{\mathbf{x}\in J_k}dx_k\hdots dx_1\geq r^k.\]
    Now, consider an arbitrary  $\mathbf{x}\in J_k$ such that $\mathbf{x}=(x_1,\hdots,x_k, y_{k+1},\hdots, y_n)$, and the line     
    $$\pline_{(x_1,\hdots,x_k, y_{k+1},\hdots, y_n),\bm e_{k+1}}\triangleq\{(x_1,\hdots,x_k,y_{k+1}+t, y_{k+2},\hdots, y_n):t\in\R\}. $$
    Again, by \autoref{cor:local-sandwich-improved}, there exists $J_{k+1}^{(x_1,\hdots,x_k)}\subset \pline_{(x_1,\hdots,x_k, y_{k+1},\hdots, y_n),\bm e_{k+1}} \cap\cube$, such that,  
    for every $$\bm z=( x_1,\hdots, x_k, z_{k+1},y_{k+2},\hdots, y_n)\in J_{k+1}^{(x_1,\hdots,x_k)}$$ 
    we have, 
    \begin{align*}
        |p(\bm z)|=|p(x_1,\hdots,x_k, z_{k+1},y_{k+2},\hdots,y_n)|\geq \frac{|p(x_1,\hdots,x_k, y_{k+1},\hdots, y_n)|}{2}\geq\frac{|p(\bm y)|}{2^{k+1}},
    \end{align*}
    and
    \begin{equation}\label{eq:line-int}
        \left|J_{k+1}^{(x_1,\hdots,x_k)}\right|=\int_{\bm z\in J_{k+1}^{(x_1,\hdots,x_k)}}dz_{k+1} \geq r.
    \end{equation}
    Observe, for any $(x_1,\hdots,x_k)\neq(u_1,\hdots,u_k)$, $\pline_{(x_1,\hdots,x_k, y_{k+1},\hdots, y_n),\bm e_{k+1}} \cap \pline_{(u_1,\hdots,u_k, y_{k+1},\hdots, y_n),\bm e_{k+1}} = \emptyset$. 
    So, we may deduce $J_{k+1}^{(x_1,\hdots,x_k)} \cap J_{k+1}^{(u_1,\hdots,u_k)} =\emptyset$. So, for every $\mathbf{x}\in J_k$, the corresponding line segment along $\bm e_{k+1}$ is unique.
        Note that, for any $\mathbf{x}\in J_{k}$, from \eqref{eq:line-int} we may infer:
    \begin{equation}\label{eq:line-translation}
        \int_{\bm z\in J_{k+1}^{(x_1,\hdots,x_k)}}dz_{k+1}\geq \int_{0}^{r}dz_{k+1}=r.
    \end{equation}
    So, all the line segments can be translated along $\bm e_{k+1}$, to cover the interval $[0,r]$. \autoref{fig:translation} illustrates this idea, which allows us to construct higher dimensional cuboids.
    \begin{figure}
        \centering
        \includegraphics[scale=0.45]{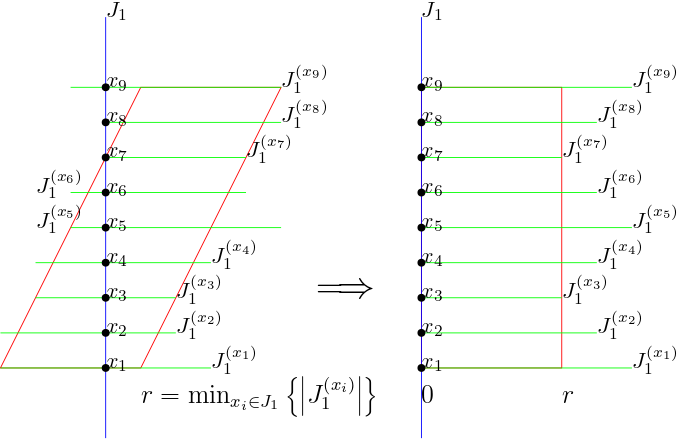}
        \caption{Translation for constructing $J_2$: For every $\mathbf{x}_i\in J_1$ (blue), the line segments $J_1^{(\mathbf{x}_i)}$ (green) are translated to have their left ends at $0$, so that they all cover the rectangle (red) of height $|J_1|\geq r$, and width $r$. Thus, $|J_2|\geq r^2$.}
        \label{fig:translation}
    \end{figure}
        We define
    \[J_{k+1}\triangleq\bigcup_{(x_1,\hdots,x_k,y_{k+1},\hdots, y_n)\in J_k} J_{k+1}^{(x_1,\hdots,x_k)}.\]
    Since, all the line segments $J_{k+1}^{(x_1,\hdots,x_k)}$ are unique, we have:
    \begin{align*}
        |J_{k+1}|&\triangleq\int_{\mathbf{v}=(x_1,\hdots,x_{k+1},y_{k+2},\hdots,y_n)\in J_{k+1}}dx_{k+1}\hdots dx_1\\
        &=\int_{\mathbf{x}\in J_k} \int_{\bm z\in J_{k+1}^{(x_1,\hdots,x_k)}} dz_{k+1}
                dx_k \hdots dx_1 \\
        &\geq\int_{\mathbf{x}\in J_k} \underbrace{\int_{0}^{r} dz_{k+1}}_{= r}dx_k\hdots dx_1 \tag{From \eqref{eq:line-translation}}\\
                &\geq r\underbrace{\int_{\mathbf{x}\in J_k}dx_k \hdots dx_1}_{=|J_k|\geq r^k}\geq r^{k+1}.
    \end{align*}
    Thus, there exists $ J=J_n\subset\cube\text{, such that }|J|\geq r^n$, and for every $\mathbf{x}\in J,|p(\mathbf{x})|\geq\frac{|p(\bm y)|}{2^n}$.
\end{proof}

The main theorem of this section now easily follows from \autoref{lem:gen-prod-set}. 

\begin{proof}[Proof of \autoref{thm:improved-sandwich}] Let $\mathbf{x}^*\in\cube\text{, such that }|p(\mathbf{x}^*)|=\|p\|_{\cube_n,\infty}$. By \autoref{lem:gen-prod-set}, we have a $J\subset\cube$, such that $|J|\geq \frac{1}{(2d^2)^n}$, and for every $\mathbf{x}\in J,|p(\mathbf{x})|\geq\frac{|p(\mathbf{x}^*)|}{2^n}$. Then,
\[\|p\|_{\cube_n,1} =\int_{\cube}|p(\mathbf{x})|d\mathbf{x} \geq\int_{J}|p(\mathbf{x})|d\mathbf{x} \geq\int_{J}\frac{|p(\mathbf{x}^*)|}{2^n}d\mathbf{x} =\frac{|p(\mathbf{x}^*)|}{2^n}|J| \geq\frac{\|p\|_{\cube_n,\infty}}{(2d)^{2n}}.\qedhere\]   
\end{proof}

\subsection{Tightness of \autoref{thm:improved-sandwich}}
We note that the lower bound of $\frac{1}{(2d)^{2n}}$ for $\frac{\|p\|_{\cube_n,1}}{\|p\|_{\cube_n,\infty}}$ is asymptotically tight, complementing \autoref{thm:improved-sandwich}:
\tightnessofinduction*
A simple observation towards this claim, may first be noted:
\begin{obs}
    There exists a family of a individual degree-$d$ polynomials $\{h_n\}_{n\in\mathbb N}$, where $h_n:\cube_n\to\R$, such that $\|h_n\|_{\cube_n,\infty}\geq (d/2)^n \|h_n\|_{\cube_n,1}$. Define $h_n(\bm x)\triangleq\prod_{i=1}^n x_i^d$, and observe 
    $\|h_n\|_{\cube_n,\infty}=
        1$, while
    \[\|h_n\|_{\cube_n,1}=\int_{\bm x\in\cube_n}\prod_{i=1}^n |x_i^d| d\bm x =\prod_{i=1}^n\int_{-1}^1|x_i^d|dx_i=\left(2\int_0^1 x^d dx\right)^n=\left(\frac{2}{d+1}\right)^n \leq \left(\frac{2}{d}\right)^n \|h\|_{\cube_n,\infty}.\]
    \end{obs}
A tighter bound maybe shown using Legendre polynomials, as pointed out by \cite{legendre-bound}:
\begin{defn}\label{defn:legendre-poly}
    Legendre Polynomials are a family of orthogonal polynomials $\{P_n:[-1,1]\to[-1,1]\}_{n\in\mathbb N}$, indexed by their degree $n$, such that $P_0(x)=1$, for all $x\in[-1,1]$, $P_n(1)=1$, for all $n$, and for all $m\neq n$,
    \begin{equation}\label{eq:legendre-poly-ortho}
        \int_{-1}^1 P_n(x)P_m(x)dx=0.
    \end{equation}
\end{defn}
We note some useful properties of Legendre polynomials:
\begin{fact}\label{fact:legendre-poly}
    Legendre Polynomials are even or odd, according as, their degree is even or odd, i.e.
    \begin{equation}\label{eq:legendre-poly-parity}
        P_n(-x)=(-1)^n P_n(x).
    \end{equation}
    Accordingly their derivatives, $P'_n:[-1,1]\to\R$, follow a similar rule:
    \begin{equation}\label{eq:legendre-poly-derivatives-parity}
        P'_n(-x)=(-1)^{n+1} P'_n(x).
    \end{equation}
    Their derivatives attain their respective extrema at the end points of the $[-1,1]$ interval:
    \begin{equation}\label{eq:legendre-poly-derivatives-extrema}
        \|P'_n\|_{[-1,1],\infty}=P'_n(1)=\frac{n(n+1)}{2}\text{,\qquad and \qquad}P'_n(-1)=(-1)^{n+1}\frac{n(n+1)}{2}.
    \end{equation}
    For any $n\in\mathbb{N}$, as any polynomial $f(x)$ of degree $m<n$ can be written as a real linear combination of $\{P_i(x)\}_{i=0}^m$, i.e. $f(x)\equiv\sum_{i=0}^m c_iP_i(x)$ for some $\{c_i\in\R\}_{i=0}^m$, using \eqref{eq:legendre-poly-ortho} we get
    \begin{equation}\label{eq:legendre-poly-f-ortho}
        \int_{-1}^1 P_n(x)f(x)dx=\sum_{i=0}^{m} c_i \underbrace{\int_{-1}^1 P_n(x)P_i(x)dx}_{=0,\because i\neq n}=0.
    \end{equation}
\end{fact}
\begin{remark}[\cite{legendre-bound}]\label{remark:stack-exchange}
    Let $m\in\mathbb{N}$. Consider the Legendre polynomial $P_{m+1}$, and define a degree-$(2m+1)$ polynomial $p:[0,1]\to\R$ as $p(x)\triangleq x(P'_{m+1}(2x-1))^2$. Then, $\|p\|_{[0,1],\infty}\geq (m+1)(m+2)\|p\|_{[0,1],1}$.
\end{remark}
\begin{proof}
    From \eqref{eq:legendre-poly-derivatives-extrema}, we have $\max_{x\in[0,1]}\{ |P'_{m+1}(2x-1)|\} =\|P'_{m+1}\|_{[-1,1],\infty} =P'_{m+1}(1)$. So, we observe:
    \begin{itemize}
        \item $p(x)\geq 0$ for every $x\in[0,1]$, and
        \item $\|p\|_{[0,1],\infty}=\max_{x\in[0,1]}\{x(P'_{m+1}(2x-1))^2\}=p(1)=(P'_{m+1}(1))^2$.
    \end{itemize} 
    Since $P'_{m+1}$ is a degree-$m$ polynomial, for some degree-$(m-1)$ polynomial $q:\R\to\R$, we may write
    \begin{equation}\label{eq:legendre-poly-decomposition}
        P'_{m+1}(2x-1)=P'_{m+1}(1)+(x-1)q(2x-1).
    \end{equation}
        \begin{align*}
        \implies \|p\|_{[0,1],1} &= \int_{0}^1 x(P'_{m+1}(1)+(x-1)q(2x-1))P'_{m+1}(2x-1) dx \tag{By \eqref{eq:legendre-poly-decomposition}}\\
        &= P'_{m+1}(1)\underbrace{\int_0^1 x P'_{m+1}(2x-1)dx}_{\triangleq I_1} +\underbrace{\int_0^1 x(x-1)q(2x-1)P'_{m+1}(2x-1)dx}_{\triangleq I_2}.
            \end{align*}
    Using integration by parts on the second integral, we get
    \begin{align*}
        I_2 &=\int_0^1 x(x-1)q(2x-1)P'_{m+1}(2x-1)dx \\
        &= \cancelto{0}{\Bigg[\underbrace{x(x-1)q(2x-1)\int P'_{m+1}(2x-1)dx}_{=(x(x-1)q(2x-1)P_{m+1}(2x-1))/2}\Bigg]_0^1} -\int_0^1 P_{m+1}(2x-1) \underbrace{\frac{d(x(x-1)q(2x-1))}{dx}}_{\triangleq r(2x-1),\; deg(r)\leq m}dx \\
        &=-\int_{0}^1 P_{m+1}(2x-1)r(2x-1)dx =-\frac{1}{2}\int_{-1}^1 P_{m+1}(y)r(y)dy = 0. \tag{By \eqref{eq:legendre-poly-f-ortho}}
    \end{align*}
    Using integration by parts, this time on the first integral, we get
    \begin{align*}
        I_1 &=\int_0^1 x P'_{m+1}(2x-1)dx = \left[x\int P'_{m+1}(2x-1)dx\right]_0^1 - \int_0^1 \left(\int P'_{m+1}(2x-1)dx\right) dx \\
        &=\left[\frac{x}{2} P_{m+1}(2x-1)\right]_0^1 - \frac{1}{2}\int_0^1  P_{m+1}(2x-1) dx =\frac{P_{m+1}(1)}{2}-\frac{1}{4}\underbrace{\int_{-1}^1  P_0(y)P_{m+1}(y) dy}_{=0,\; by\; \eqref{eq:legendre-poly-ortho}}.\\
        \implies &\|p\|_{[0,1],1} =\frac{P'_{m+1}(1) P_{m+1}(1)}{2}=\frac{P'_{m+1}(1) }{2}. \tag{$P_{m}(1)=1$, for all $m\in\mathbb N$}\\
        \implies &\frac{\|p\|_{[0,1],1}}{\|p\|_{[0,1],\infty}} =\frac{P'_{m+1}(1)}{2 (P'_{m+1}(1))^2}=\frac{1}{2 P'_{m+1}(1)}=\frac{1}{(m+1)(m+2)}. \tag{By \eqref{eq:legendre-poly-derivatives-extrema}}
    \end{align*}\qedhere
\end{proof}
Using this, we can now prove the main theorem of this sub-section:
\begin{proof}[Proof of \autoref{lem:tightness_of_induction}]
    Let $d=2m+1$ for some $m\in \mathbb{N}$. Define a degree-$(2m+1)$ polynomial $f:[-1,1]\to\R$, as $f(x)=p\left(\frac{x+1}{2}\right)$, where $p:[0,1]\to\R$ is as defined in \autoref{remark:stack-exchange}. Then:
    \begin{align*}
        \|f\|_{[-1,1],\infty} &=\|p\|_{[0,1],\infty} \text{, and } \\
        \|f\|_{[-1,1],1} &=\int_{-1}^1 |f(x)|dx = \int_{-1}^1 \left|p\left(\frac{x+1}{2}\right)\right| dx = 2\int_{0}^1 |p(y)| dy= 2\|p\|_{[0,1],1}.
    \end{align*}
    From \autoref{remark:stack-exchange}, we have
    \[\|f\|_{[-1,1],\infty} =\|p\|_{[0,1],\infty} = (m+1)(m+2)\|p\|_{[0,1],1} = \frac{(m+1)(m+2)}{2}\|f\|_{[-1,1],1}.\]
    For any $n\in\mathbb N$, define $f_n(\bm x)\triangleq \prod_{i=1}^n f(x_i)$. Observe $\|f_n\|_{\cube_n,\infty}=(\|f\|_{[-1,1],\infty})^n$, while
    \[\|f_n\|_{\cube_n,1}=\int_{\bm x\in\cube_n}\prod_{i=1}^n |f(x_i)| d\bm x =\prod_{i=1}^n \int_{-1}^1 |f(x_i)| dx_i =(\|f\|_{[-1,1],1})^n .\]
    \[\implies \frac{\|f_n\|_{\cube_n,\infty}}{\|f_n\|_{\cube_n,1}} =\left(\frac{\|f\|_{[-1,1],\infty}}{\|f\|_{[-1,1],1}}\right)^n = \left(\frac{(m+1)(m+2)}{2}\right)^n\ge \left(\frac 1 8\right)^n d^{2n}.\qedhere\]
\end{proof}

\section{Robust multivariate regression algorithm}\label{sec:l-infty-with-l-1}
In this section, we show how a modification of our algorithm that avoid a run-time dependence on $(\|p\|_{\cube_n,\infty},\sigma,N)$ and thus prove the following theorem. 
\begin{restatable}{thm}{linftywithloneupperbound}\label{cor:final-cor-with-l-1}[Generalized form of \autoref{cor:final-cor-with-l-1-intro}]
 Let $\eps\in(0,1/2]$,$\delta\in(0,\eps]$, $\sigma\geq 0$, and  $\rho<1/2$.
   There is an algorithm that solves the \problemDefn with approximation factor $C=2+\eps$. The  output of the algorithm is a polynomial $\hat{p}\colon\R^n\to\R$ of degree at most $d$ in each variable, such that with probability (over the random input samples) at least $1-\delta$, $\hat{p}$ satisfies  
   \[|p(\bfx)-\hat{p}(\bfx)|\leq (2+\eps)\sigma,\qquad \text{ for all } \bfx\in\cube.\]
 It uses   
 $M=O_{n,\rho}\left(\left(d^{2n+1}/\eps\right)^{n}\log(d/\delta)\right)$ samples drawn from the multidimensional Chebyshev distribution, or $M=O_{n,\rho}\left(\left(d^{2n+1}/\eps\right)^{2n}\log(d/\delta)\right)$ if  the sampled are drawn from the uniform measure.
 Its run-time is $\poly(M,\log_{\eps}(1-2\rho))$.
\end{restatable}

The first step is to generalize KKP's $\ell_1$ regression. Similar to their regression on averages over Chebyshev intervals, we do regression on averages over Chebyshev cells. 
\begin{defn}[$\ell_1$ Minimizer]\label{defn:l1-minimzer}
    Let $m$ be a large enough integer.     Given a set of $M$ samples $S=\{(\bfx_i,y_i)\}$, for every $\bmj\in[m]^n$, let $S_{\bmj}\triangleq\{\beta\in[M]:\bfx_{\beta}\in\cube_{\bmj}\}$. The $\ell_1$ minimizer of $S$ with respect to the $(m,n)$-Chebyshev partition, is the individual degree-$d$ polynomial:
    \begin{equation}
        \label{eq:def_l1-minimzer}
        \hat{p}_{\ell_1}\triangleq\arg\min_{f\in\Pd}
    \sum_{\bm j\in[m]^n} \frac{V_n(\cube_{\bm j})}{|S_{\bm j}|}\sum_{\beta\in S_{\bm j}}|f(\mathbf{x}_{\beta})-y_{\beta}|.
    \end{equation}
    \end{defn}

We show that on a set of $\alpha$-good samples, the $\ell_1$ regression outputs an individual degree-$d$ polynomial with $\poly(d^n)$ error in $\ell_\infty$.
\begin{thm}[$\ell_\infty$ error bound for the $\ell_1$ minimizer]\label{thm:l1-minimizer-l-infty-error}
    Let $\alpha<1/2$ be constant $,\eps\leq(1-2\alpha)/2$, and for some constant $c>1$, $m\geq(cd)^{2n+1}/\eps$. Given a set $S$ of samples     that is $\alpha$-good with respect to the $(m,n)$-Chebyshev partition, 
    the $\ell_1$ minimizer $\hat{p}_{\ell_1}$ from \eqref{eq:def_l1-minimzer}     satisfies 
    \[\|p-\hat{p}_{\ell_1}\|_{\cube_n,\infty}    =    O_{\alpha}((8d^2)^n \sigma).\]
\end{thm}
This proof is provided in \autoref{subsec:proof-of-sec:l-1-bound}.
In \autoref{alg:median_rec_with_l1}, we use the same idea as in KKP: starting with $\hat{p}_{\ell_1}$, the result of an $\ell_1$ regression, we then iteratively refine the estimate, improving the $\ell_{\infty}$ error in each step.

\begin{algorithm}[ht]
\caption{Median Based Recovery, with $\ell_1$ regression}\label{alg:median_rec_with_l1}
  \Given {A set of samples $S=\{\mathbf{x}_i,y_i\}_{i=1}^M$, of which a $\rho$ fraction may be outliers.}
  $\hat{p}^{(0)}\gets\text{ result of $\ell_1$ regression: }\hat{p}_{\ell_1}$\;
    $N_{\ref{alg:median_rec_with_l1}}\gets O\left(n\log_{ 1/\eps} d +\log_{\eps}(1-2\rho)\right) $\;
    \For{$i\in \{0,\dots,N_{\ref{alg:median_rec_with_l1}}-1\}$}{$\hat{p}^{(i+1)}\gets\emph{\Call{Refine}{$S,\hat{p}^{(i)}$}}$\;}
    \textbf{Return} $\hat{p}^{(N_{\ref{alg:median_rec_with_l1}})}$.
 \end{algorithm}

As a result, the number of iterations here depends only on $(d,n)$, and improves as  $\eps$ gets smaller. (In fact, it is linear in $n$, and logarithmic in $d$, and $1/\eps$). 

Here, as in \autoref{sec:l-infty-without-l-1}, we prove that with enough samples, the set of samples is good with high probability, and separately we show that for any $\alpha$-good set of samples \autoref{alg:median_rec_with_l1} recovers $p$ as required. 
\begin{thm}[Absolute $\ell_\infty$ error bound]\label{thm:l-infty-guarantee-with-l_1}
    Let $\eps,\alpha<1/2$. Let the set $S=\{(\mathbf{x}_i,y_i)\}_{i=1}^M$ of samples be $\alpha$-good for the $(m,n)$-Chebyshev partition 
    where $m\geq(cd)^{2n+1}/\eps$, for some large enough constant $c>1$. Then the median recovery \autoref{alg:median_rec_with_l1}, in $N_{\ref{alg:median_rec_with_l1}}=O(n\log_{1/\eps}d +\log_{\eps}(1-2\alpha))$ iterations, returns an individual degree-$d$ polynomial $\hat{p}$, such that $\|p-\hat{p}\|_{\cube_n,\infty}\leq(2+\eps)\sigma$.
\end{thm}

\begin{proof}
Using the same notations as in the proof of \autoref{thm:l-infty-guarantee}, we conclude 
  \begin{align*}
        \|p-\hat{p}^{(t)}\|_{\cube_n,\infty}&\leq(2+6\eps)\sigma+\eps^t\|e_0\|_{\cube_n,\infty}.
    \end{align*}
    In  \autoref{alg:median_rec_with_l1},  $\hat{p}^{(0)}$ is set to be the result of $\ell_1$ regression, giving us $e_0=p-\hat{p}_{\ell_1}$. Since $S$ is assumed to be $\alpha$-good for the $(m,n)$-Chebyshev partition, where $m=\Omega((cd)^{2n+1}/\eps)$, for some absolute constant $c>1$, applying  \autoref{thm:l1-minimizer-l-infty-error}, we get $\|e_0\|_{\cube_n,\infty}\leq (8d^2)^{n} O_{\alpha}(\sigma)$.  
                        Thus, 
    \begin{align*}
        \|p-\hat{p}^{(t)}\|_{\cube_n,\infty}\leq(2+6\eps)\sigma +\eps^t(2\sqrt{2}d)^{2n}\frac{4\sigma}{1-2\alpha}.
    \end{align*}
    So, in $N_{\ref{alg:median_rec_with_l1}}\geq \log_{1/\eps}\frac{(2\sqrt{2}d)^{2n}}{1-2\alpha}$ iterations, followed by a further rescaling of $\eps$ to $\eps/10$, we get the desired error bound.
\end{proof}

Now we are ready to prove the main theorem of this section:

\begin{proof}[Proof of \autoref{cor:final-cor-with-l-1}]
    It follows the same arguments as the proof of \autoref{thm:l_infty_regression_upper_bound}, with \autoref{thm:l-infty-guarantee} being replaced by \autoref{thm:l-infty-guarantee-with-l_1}, according to which: for \autoref{alg:median_rec_with_l1}, $m=(cd)^{2n+1}/\eps$ suffices, for some absolute constant $c>1$. This gives us the Chebyshev, and Uniform sample complexities of     $M_C=\frac{1}{(1-2\rho)^2}\left(\frac{cd^{2n+1}}{\eps}\right)^n\log\frac{d}{\eps\delta}$, and
    $M_U=\frac{1}{(1-2\rho)^2}\left(\frac{cd^{2n+1}}{\eps}\right)^{2n}\log\frac{d}{\eps\delta}$, respectively, for some constant $c>1$.
                
    The value of $\alpha$ that we set to deal with $\rho(< 0.5)$ fraction of outliers remains the same: $\alpha=\frac{2\rho+1}{4}<\frac{1}{2}$. So, the run-time is that of solving $N_{\ref{alg:median_rec_with_l1}}=O(n\log_{1/\eps}d +\log_{\eps}(1-2\rho))$ linear programs with $O(d^n)$ variables, and $M$ constraints. 
\end{proof}

\section{Sample Complexity Lower Bounds}\label{sec:lower-bound}
\subsection{Sample complexity lower bound against uniform sampling}\label{subsec:uniform-lower-bound}
In this section, we prove \autoref{thm:lower-bound-intro}, restated here for convenience:
\lowerboundintro*
Before we prove \autoref{thm:lower-bound-intro}, we formally note the following definitions:
\begin{defn}[Chebyshev Polynomials]\label{defn:chebyshev-polynomial}
    Chebyshev polynomials of the first kind are degree-$d$ polynomials $T_d:\R\to\R$, that follow the recurrence relation:
    \begin{align*}
        T_0(x)=1,\qquad
        T_1(x)=x,\qquad
        T_{d+1}(x)=2xT_d(x)-T_{d-1}(x).
    \end{align*}
    \[T_d(x)\triangleq\begin{cases}
    \cos(d\arccos(x)),&\text{if } |x|\leq 1,\\
    \cosh(d\arccosh(x)),&\text{if } x\geq 1,\\
    (-1)^d\cosh(d\arccosh(-x)),&\text{if } x\leq -1,
\end{cases}\text{ is their explicit trigonometric formulation.}\]
\end{defn}

\begin{defn}[Chebyshev Extremas]\label{defn:chebyshev-extremas}
    For any $d\in\Z_{>0}$, the $d$ Chebyshev extremas $\in[-1,1]$ given by
    \[x_k\triangleq\cos\left(\frac{k}{d}\pi\right),k\in[d]\]
    are the extremas of $T_d$, the degree-$d$ Chebyshev polynomial of the first kind, i.e., $T_d(x_k)\in\{\pm 1\},\forall k\in[d]$.
\end{defn}

\begin{proof}[Proof of \autoref{thm:lower-bound-intro}]
              
   For $n=1$, KKP showed that the polynomial $T_d\left(x_1+\frac{\alpha}{d^2}\right)$ (for some constant $\alpha$ dependent on $C$) 
   and the identically $0$ polynomial are indistinguishable with high probability. 
   We build on their result.  Let $\alpha= 4\sqrt{C}$ and consider the following  multivariate polynomial, with degree $d$ in each variable:
           \[f(\mathbf{x})\triangleq\frac{\prod_{i=1}^n T_d\left({x}_i+\frac{\alpha}{d^2}\right)}{\left(T_d\left(1+\frac{\alpha}{d^2}\right)\right)^{n-1}}.\]
     Let the target polynomial $p$ to be learned be  $f$ with probability $1/2$ and otherwise,  $g\equiv 0$.
Then,
\begin{itemize}
    \item[(i)] For every $\mathbf{x}\in\cube_n\setminus\left[1-\frac{\alpha}{d^2},1\right]^n$ there exists at least one index $i\in[n]$ 
 such that $x_i+\frac{\alpha}{d^2}\le 1 $, implying $\left|T_d\left(x_i+\frac{\alpha}{d^2}\right)\right|\le 1$. For all other $j\neq i\in[n]$, we have $\left|\frac{T_d\left(x_j+\frac{\alpha}{d^2}\right)}{T_d\left(1+\frac{\alpha}{d^2}\right)}\right|\leq 1$, since $T_d(1)=1$, and $T_d(x)$ is monotonically increasing on the region $x\geq 1$. Hence, $|f(\mathbf{x})-g(\mathbf{x})|=|f(\mathbf{x})|\leq 1$. We thus say that $f$ and $g$ are indistinguishable on the region $\cube_n\setminus\left[1-\frac{\alpha}{d^2},1\right]^n$ for noise level $\sigma=1$. We assume $d>\sqrt{\alpha/2}$,  as otherwise, the theorem follows trivially for $c=\sqrt{2/\alpha}$. 
 Thus, the probability that $\bfx$ lands in the \emph{distinguishable} region is $=\frac{V_n\left([1-\alpha/d^2,1]^n\right)}{V_n(\cube_n)}=\frac{1}{2^n}\left(\frac{\alpha}{d^2}\right)^n$.
    \item[(ii)] For all $d>1$ we have $\|f-g\|_{\cube_n,\infty}> 2C$. 
    This follows from 
     \begin{align*}
        \|f-g\|_{\cube_n,\infty}\geq|f(\bm 1)-g(\bm 1)|=\left|T_d\left(1+\frac{\alpha}{d^2}\right)\right|>\frac{\alpha^2}{8}=2C.
    \end{align*}
    For all $\alpha> 0$, the  second inequality can be verified for $d=2$ (where  $T_2(x)=2x^2-1$), and follows for all $d>2$ since  $T_d\left(1+\frac{\alpha}{d^2}\right)=\cosh\left(d\arccosh\left(1+\frac{\alpha}{d^2}\right)\right)$ is increasing in $d$, as KKP noted. 
\end{itemize}
 Then, for $M< (cd)^{2n}$ independent samples uniformly generated over $\cube$, the event that all of them land in the {indistinguishable} region, happens with probability:
\begin{align*}
    \Pr[\forall i\in[M].\ \mathbf{x}_i\in\cube_n\setminus[1-\alpha/d^2,1]^n]&\ge 1-M\cdot\Pr[\mathbf{x}\not\in\cube_n\setminus[1-\alpha/d^2,1]^n]\\
    &=1-M\left(\frac{\alpha}{2d^2}\right)^n> \frac 2 3 .
                \end{align*}
The first inequality is by union bound over the $M$ samples and the last inequality follows from $M< \frac{1}{3}\left(\frac{2d^2}{\alpha}\right)^{n}$ (e.g., for $c=(36C)^{-1/4}$). 
Consider the following adversarial strategy: the adversary chooses $p'=f$ with probability $1/2$, and otherwise $p'=g$. For samples in the indistinguishable region, it outputs $p'$, and otherwise $p$. 
Now,  with probability $>2/3$, all the samples are from the indistinguishable region. So, with this probability, the algorithm observes only values of $p'$, which are independent of the target function $p$, and the best it can do is to guess the function with probability $1/2$. Thus, it fails with probability $>2/3\cdot 1/2=1/3$.
\end{proof}

\begin{remark}\label{rem:lower-bound-total}
We can get a lower bound of $(c d/n)^{2n}$
(with the same value of $c$ as in \autoref{thm:lower-bound-intro}) for the sample complexity of learning {\em total degree-$d$} polynomials in our setting, by simply applying the above argument with the individual degree being $\lfloor d/n\rfloor$. One can get a slightly improved bound of $(c d/\sqrt{n})^{2n}$ for this problem by using the polynomial $T_d\left(\frac{\alpha}{d^2}+\frac1n \sum x_i\right)$, with the same values of $\alpha$ and $c$ as above. Here, the indistinguishable region becomes $\cube_n\setminus \left[1-n\alpha/d^2, 1\right]^n$, and so, the analysis slightly changes. 
\end{remark}

\subsection{Distribution-free sample complexity lower bound}\label{subsec:dist-free-lower-bound}
In this section, we prove \autoref{thm:dist-free-lower-bound-intro}. We start with some auxiliary lemmas. 
We begin by generalizing \cite[Lemma 5.2]{kkp}, which we state for reference:
\begin{lem}[\cite{kkp}, Lemma 5.2]\label{lem:kkp-dist-free-aux-1}
    Let $d\geq 1$. For any point $a\in[-1,1]$, there exists a degree-$d$ polynomial  $q_a:\R\to\R$, such that $q_a(a)=1=\|q_{a}\|_{[-1,1],\infty}
    $, and for every $x\in[-1,1]$,
    \[|q_{a}(x)|\leq \frac{2}{d|x-a|}.\]
\end{lem}

Next, we generalize \cite[Lemma 5.3]{kkp}, a variant of which we state for reference, and then briefly describe their argument:

\begin{lem}[\cite{kkp}, Lemma 5.3]\label{lem:kkp-dist-free-aux-2-fixed}
    For any $d,\alpha>0$, let $m=\lfloor d\alpha/2\rfloor$. Define a set of nodes $b_j\triangleq -1+\frac{2j}{m},j\in[m]$. For any $S\subseteq[m]$, consider the set of degree-$d$ polynomials 
    \[g_{S}(y)\triangleq\sum_{j\in S} q_{b_j}(y),\] 
    where $q_{b_j}:[-1,1]\to\R$ is the degree-$d$ polynomial guaranteed from \cite[Lemma 5.2]{kkp}. For any $y\in[-1,1]$, define $k_{y}\triangleq\arg\min_{j\in[m]} |b_{j}-y|$, i.e. the index of the node $b_j$ closest to $y$. Then, for any $j\in [m]$, if $j\neq k_{y}$ then $$|g_{\{j\}}(y)|\leq  \alpha.$$ 
        \end{lem}

Now, for the (univariate) lower bound construction (in short): the adversary
\begin{itemize}
    \item picks a random $S\subseteq [m]$, and an arbitrary $j^*$ out of the set of outlier full nodes (i.e. those nodes for which all the sample points, that these nodes are the nearest nodes, are outliers),
    \item defines $S'\triangleq S\Delta\{j^*\}$ (i.e. $S'=S\cup \{j^*\}$, if $\{j^*\}\not\in S$, and  $S'=S\setminus \{j^*\}$, if $\{j^*\}\in S$) so that $S\Delta S'=\{j^*\}$, and
    \item replies with either $f_S$ or $f_{S'}$ with probability $1/2$.   
\end{itemize}
  
Note, for every sample $x$ we have $|f_S(x)-f_{S'}(x)|=|f_{\{j^*\}}(x)|$. Either
\begin{itemize}
    \item $k_x=j^*$ i.e. $x\in L_{j^*}$ is outlier (in which case $|f_S(x)-f_{S'}(x)|$ is arbitrary), or
    \item $k_x\neq j^*$, implying $|x-b_{j^*}|\ge 1/m$,  and hence, by \autoref{lem:kkp-dist-free-aux-2-fixed}
\[|f_S(x)-f_{S'}(x)|=|f_{\{j^*\}}(x)|=|q_{b_j^*}(x)|\leq \frac{2}{d|x-b_{j^*}|}\leq\alpha.\]
\end{itemize} 
The last inequality follows for $m=d\alpha/2$.
We now generalize \autoref{lem:kkp-dist-free-aux-1}.
\begin{lem}\label{lem:dist-free-lb-aux-1}
    Let $d\geq 1$. For any $\bm b=(b_1,\hdots,b_n)\in\cube$, there exists a polynomial $p_{\bm b}:\cube\to\R$ with individual degree at most $d$, such that $|p_{\bm b}(\bm b)|=1    $, 
    and for all $\bm x=(x_1,\hdots,x_n)\in\cube$,        \[|p_{\bm b}(\bm x)|\leq\frac{2}{d\max_{i\in[n]}| x_i - b_i|}.\]
\end{lem}

\begin{proof}[Proof of \autoref{lem:dist-free-lb-aux-1}]
    Fix some $\bm b\in\cube$. 
                                 Consider the individual degree at most $d$ polynomial $p_{\bm b}:\cube\to\R$ defined as:
                
        \begin{equation}\label{eq:defn-of-pb-mult}
        p_{\bm b}(\bm x)\triangleq\prod_{i=1}^n q_{b_i}(x_i).
    \end{equation} 
    Where $q_{b_i}$ are the polynomials defined in \autoref{lem:kkp-dist-free-aux-1}.
    Since for all $i\in[n]$, $|q_{b_i}(b_i)|=1    $,   we have,
    \[|p_{\bm b}(\bm b)|=\prod_{i=1}^n |q_{b_i}(b_i)| = 1     .\]
    We next show that $|p_{\bm b}(x)|$ cannot be too large for all $\bm x\in \cube$. Let $i'$ be the index that achieves the maximum of $|x_i-b_i|$. For every $i\neq i'$ we use the bound $|q_{b_i}(x_i)|\leq 1$, while for $i'$ we use the bound $|q_{b_i}(x_i)|\leq \frac{2}{d|x_i-b_i|}$, with both the bounds being from \autoref{lem:kkp-dist-free-aux-1}.
      We get,
    \[        |p_{\bm b}(\bm x)|=\prod_{i=1}^n |q_{b_i}(b_i)|\leq\frac{2}{d\max_{i\in[n]}| x_i - b_i|}.\qedhere
    \]

\end{proof}

In order to generalize \autoref{lem:kkp-dist-free-aux-2-fixed}, let us define \emph{nodes}, and the notion of \emph{closest} nodes in the cube $\cube$:
\begin{defn}\label{defn:closest-nodes}
    For any $m>0$, consider the set of points $b_j\triangleq -1+\frac{2j}{m},j\in[m]$, that equi-partition $[-1,1]$ along each axes. For any $\bm j=(j_1,\hdots,j_n)\in[m]^n$, the corresponding \emph{node} is $\bm b_{\bm j}\triangleq (b_{j_1},\hdots,b_{j_n})$.
    
    For any $\mathbf{x}=(x_1,\hdots,x_n)\in\cube$, the closest (in $\ell_1$)  node to $\mathbf{x}$ is $\bm b_{\bm k (\mathbf{x})}$, where      $ \bm k (\mathbf{x})\triangleq(k_1,\hdots,k_n)$, and $k_i\triangleq\arg\min_{j\in[m]}|b_j-x_i|$,
     for all $i\in[n]$.
\end{defn}

\begin{lem}\label{lem:dist-free-aux-2}
    For any $d,\alpha>0$, let $m=\lfloor d\alpha/2\rfloor$. 
        Consider the set of $m^n$ \emph{nodes} $\{\bm b_{\bm j},\bm j\in[m]^n\}$, and the notion of \emph{closest nodes}, as discussed in \autoref{defn:closest-nodes}.
    For any subset of nodes, $\bm S\subseteq[m]^n$, define
    \[f_{\bm S}(\bm x)\triangleq\sum_{\bm j\in\bm S} p_{\bm b_{\bm j}}(\bm x),\] 
    where $\bm b_{\bm j}=(b_{j_1},\hdots,b_{j_n})$, for $\bm j=(j_1,\hdots,j_n)$ , and $p_{\bm b_{\bm j}}$'s are individual degree-$d$ polynomials from \autoref{lem:dist-free-lb-aux-1}. 
    For any $\mathbf{x}\in\cube$, let the closest in $\ell_1$  node to $\mathbf{x}$ be $\bm b_{\bm k (\mathbf{x})}$.                Then, for any $\bm j\neq \bm k (\mathbf{x})$, 
     $|f_{\{\bm j\}}(\mathbf{x})|\leq\alpha$.
    
\end{lem}

\begin{proof}[Proof of \autoref{lem:dist-free-aux-2}]
    Fix $\bm x\in \cube$, and    $\bm j\neq \bm k (\mathbf{x})$.  Then
\begin{align}\label{eq:bound-f_j}
    |f_{\{\bm j\}}(\bm x)|=|p_{\bm b_{\bm j}}(\bm x)|\leq \frac{2}{d\max_i| x_i - b_{ j_i}| }\leq\alpha, 
\end{align}
The first inequality is by \autoref{lem:dist-free-lb-aux-1} and the second holds for 
 $\max_{i\in[n]} | x_i- b_{ j_i}|\geq 1/m$. Indeed, since  $\bm j\neq \bm k (\mathbf{x})$, there exists an index $\tau\in[n]$ for which $j_{\tau}\neq k_{\tau}$, and $x_{\tau}$ is closer to $ b_{k_{\tau}}$ rather than $ b_{j_\tau}$.  So,
$
    | x_{\tau}- b_{j_{\tau}}|\geq |b_{k_{\tau}}-b_{j_{\tau}}|/2\geq 1/m. $
\end{proof}

Finally we prove the lower bound for polynomials of individual degrees at most $d$ (\autoref{thm:dist-free-lower-bound-intro}, restated below).

\distfreelowerboundintro*
\begin{proof}
We will prove a stronger bound, showing that $M<(cd)^n n\log d$ samples are not enough to succeed with constant probability.  We may assume $d>12C$ as otherwise the lower bound is trivial. 
Let $\sigma=\frac{1}{3C}$, $m=
\lfloor\frac{d}{6C}\rfloor$, and assume $M<(cd)^n n\log d$.

Consider the set of nodes $\bm b_{\bm j}$ for all $\bm j\in [m]^n$ as defined in 
 \autoref{defn:closest-nodes}. 
    For every  $\bm j=(j_1,\hdots,j_n)\in[m]^n$, 
    let   $L_{\bm j}$ be the set of samples for which the nearest node is $\bm b_{\bm j}$ in $\ell_1$ distance.
    Formally, recall that for all $\mathbf{x}\in \cube$, $\bm k(\mathbf{x})_i\triangleq \arg\min_{j}|b_j-x_i|$ where the minimization is over the one dimensional nodes $\{b_j\}_{j\in [m]}$, then
    \[L_{\bm j}\triangleq\{\beta\in[M]:\bm k(\mathbf{x}_\beta)=\bm j    \}.\]
    
    We say that $L_{\bm j}$ is \emph{outlier-full} if all the samples $(\mathbf{x}_{\beta},y_{\beta})\in L_{\bm j}$ are outliers.      We say that $L_{\bm j}$ is \emph{small}, if $|L_{\bm j}|\leq    \frac{2M}{m^n}    $.
    So, a \emph{small} $L_{\bm j}$ is \emph{outlier-full} with probability 
    \[\rho^{|L_{\bm j}|}\ge \rho^{2(cd/m)^n n\log d}    \geq d^{-o(n)}.\]
    The last inequality is by $m\ge \frac{d}{12C}$ for $d\ge 12C$ and holds for large enough $c=c(C,\rho)$.
        Since there are at least $    m^n/2$ \emph{small} $L_{\bm j}$'s,  
    the probability that none of these \emph{small} $L_{\bm j}$'s is \emph{outlier-full} is at most 
    \[(1-d^{-o(n)})^{m^n/2}\leq e^{-m^n d^{-o(n)}/2}    < 1/3.\] 
    The first inequality is by $1-x\leq e^{-x}$ for all $x$ and the last inequality holds for $d>12C$. 
    So, the probability that at least one of the \emph{small} $L_{\bm j}$'s is \emph{outlier-full} is  $\geq 2/3$.
    
    Let  $f_{\bm S}$ be the true polynomial to be learned, defined as in \autoref{lem:dist-free-aux-2}, for $\bm S$ chosen uniformly from $\mathcal{P} \left([m]^n\right)$. 
    Consider the following adversarial strategy:     
    given the samples and the outliers' locations, suppose there exists an \emph{outlier-full} region, then the adversary picks an arbitrary $\bm j^*$ such that $L_{\bm j^*}$ is \emph{outlier-full}. 
Let  $\bm S'\triangleq\bm S\Delta\{\bm j^*\}$, the adversary chooses $p'=f_{\bm S}$ or $p'=f_{\bm S'}$ equiprobably, and replies $(\mathbf{x}_{\beta}, p'(\mathbf{x}_{\beta}))$, for all $\beta\in[M]$.
Since $\bm S\Delta\bm S'=\{\bm j^*\}$, we have for all $\mathbf{x}\in\cube$, 
$|f_{\bm S}(\mathbf{x}) - f_{\bm S'}(\mathbf{x})|=|f_{\{\bm j^*\}}(\mathbf{x})|$.
For any sample $\beta\in[M]$, if this sample is not an outlier, then $\beta\not\in L_{\bm j^*}$ i.e. $\bm k(\mathbf{x}_\beta)\neq \bm j^*$, and by \autoref{lem:dist-free-aux-2}, we have  $|f_{\bm S}(\mathbf{x}) - f_{\bm S'}(\mathbf{x})|=|f_{\{\bm j^*\}}(\mathbf{x})|\leq \sigma$.
Since $\bm j^*$ is chosen based on the locations of outliers, but independent of $\bm S$, the distributions of $\bm S'$ and $\bm S$ are the same. This makes $f_{\bm S}$ \emph{indistinguishable} from $f_{\bm S'}$. We finally note, 
\[\|f_{\bm S}-f_{\bm S'}\|_{\cube_n,\infty}=\|f_{\{\bm j^*\}}\|_{\cube_n,\infty}=\|p_{\bm b_{\bm j^*}}\|_{\cube_n,\infty}\geq |p_{\bm b_{\bm j^*}}(\bm b_{\bm j^*})|=1>2C\sigma ,\]
    so that, any output of the algorithm at least $C\sigma$ far in $\ell_\infty$ from either $f_{\bm S}$, or $f_{\bm S'}$. With probability at least $2/3$ there exists a \emph{small outlier-full} $L_{\bm j^*}$, in this case, the best the algorithm can do is to guess and so the failure probability becomes at least $2/3\cdot 1/2=1/3$.
\end{proof}

\bibliographystyle{alpha}
\bibliography{refs} 

\appendix

\section{Exponential lower bound for linear functions}\label{subsec:lower-bound-linear}

\begin{restatable}{thm}{lowerboundlinear}
\label{thm:lower_bound_linear_fucntions}
For any constant approximation factor $C>1$, and any  $\sigma<\frac{1}{2C}$, given $M=e^{o(n\sigma^2)}$     samples, drawn from any product distribution with mean $0$, no algorithm can solve the \problemDefn  with failure probability  $\delta<1/4$,
    for any $\rho$ (even for $\rho=0$, i.e., even without outliers).

    In particular, for constant noise level  $\sigma$, any algorithm requires, $e^{\Omega(n)}$ many samples to succeed with probability at least $3/4$.
\end{restatable}

\begin{proof}[Proof of \autoref{thm:lower_bound_linear_fucntions}]
    Consider the two linear functions: $g\equiv 0$ ,and $h(\bfx)=\frac 1 n \sum_{i=1}^nx_i$.
    On one hand, we have $\max_{\bfx\in[-1,1]^n}|g(\bfx)-h(\bfx)|\geq|g(\bm 1)-h(\bm 1)|\ge 1$. 
        On the other hand, for many $\bfx\in\cube$,  $|g(\bfx)-h(\bfx)|\le \sigma$ holds. We call such $\bfx$'s \emph{bad} samples since the values of $h$, and $g$ are within the $\sigma$ noise limit, and thus the output function $\hat{p}$ can be close to both $h$ and $g$. 
    If all the samples are bad, the algorithm will not be able to distinguish between $g$ and $h$, so $\hat{p}$ is independent of the given samples. Since  $g$ and $h$ are at least  $1>2C\sigma$ far from each other, $\hat{p}$ can well approximate at most one of them, which success with probability at most $1/2$ --- the best the algorithm can do is to guess $h$, or $g$. 
    We next show that all the samples are bad with high probability, 
    \begin{align*}
        \Pr[|g(\bfx)-h(\bfx)|\le \sigma, \text{ for all samples}]&\ge 1-M\cdot \Pr_{\bfx}\left[\frac 1 n \sum_{i=1}^nx_i\leq\sigma\right]\ge 1-Me^{-n\sigma^2/2}> 2/3.
        \end{align*}
    The first inequality is by a union bound, the second is an application of Hoeffding's inequality.     The last inequality then follows by the assumption that $M=e^{o(n\sigma^2)}<\frac{1}{3}e^{n\sigma^2/2}$. 
\end{proof}
   We note that both the Chebyshev and the uniform distribution are product distributions with mean $0$. 
    The constraint on the distribution comes from using Hoeffding's inequality. 
    In general, \autoref{thm:lower_bound_linear_fucntions} holds for samples $\bfx$, where the coordinates $x_i$ are independent, 
    with $\Exp[x_i]=0$ for all $i$.

\section{Deferred Proofs From \autoref{sec:l-infty-without-l-1}}\label{subsec:proof-from-subsec:first-main-thm}

\begin{proof}[Proof of \autoref{lem:error_bound}]
    Let us define three piece-wise constant functions with respect to the Chebyshev partition  $r,\hat{r}$, and $\Tilde{r}$ as follows: For every  $\bm j\in[m]^n$ and for every $\mathbf{x}\in \cube_{\bm j}$,
    \[\hat{r}(\mathbf{x})=\hat{p}(\Tilde{\mathbf{x}}_{\bm j})\text{ , }\Tilde{r}(\mathbf{x})=\Tilde{y}_{\bm j}\text{, and } r(\mathbf{x})=p(\mathbf{x}'_{\bm j})\]
    where $\mathbf{x}_{\bm j}'\in \cube_{\bm j}$ is chosen such that $(\mathbf{x}_{\bm j}',\Tilde{y}_{\bm j})$ is an inlier sample. The existence of such a point $\mathbf{x}'_{\bm j}\in\cube_{\bm j}$ follows from the continuity of $p$
    and $\alpha<0.5$, i.e. more than half of the samples in $\cube_{\bm j}$ being inliers. Formally:
    \begin{claim}\label{claim:median-is-close}
        For every $\bm j\in[m]^n$, there exists a point $\mathbf{x}'_{\bm j}\in \cube_{\bm j}$, such that $|p(\mathbf{x}'_{\bm j})-\Tilde{y}_{\bm j}|\leq \sigma$.
    \end{claim}
    \begin{proof}
        Fix some $\bm j\in[m]^n$. Since, $\Tilde{y}_{\bm j}$ is the median of all the $y_i$'s, whose corresponding $\mathbf{x}_i\in\cube_{\bm j}$, there must exist $\mu,\tau\in M$, such that $\Tilde{y}_{\bm j}\in[y_{\mu},y_{\tau}]$. Also, since more than half of the samples in $\cube_{\bm j}$ are inliers, we can additionally force $(\mathbf{x}_{\mu},y_{\mu})$, and $(\mathbf{x}_{\tau},y_{\tau})$ to be inliers, i.e.
        \begin{equation*}
            |p(\mathbf{x}_{\mu})-y_{\mu}|\leq\sigma \text{, and } |p(\mathbf{x}_{\tau})-y_{\tau}|\leq\sigma.
        \end{equation*}
        So, we get $\Tilde{y}_{\bm j}\in [p(\mathbf{x}_{\mu}) +\sigma, p(\mathbf{x}_{\tau})-\sigma]$, and depending on whether $\Tilde{y}_{\bm j}$ is closer to $p(\mathbf{x}_{\mu}) +\sigma$, or $p(\mathbf{x}_{\tau})-\sigma$, we may choose $\mathbf{x}'_{\bm j}$ to be $\mathbf{x}_{\mu}$, or $\mathbf{x}_{\tau}$, respectively, ensuring $|p(\mathbf{x}'_{\bm j})-\Tilde{y}_{\bm j}|\leq \sigma$.
    \end{proof}
                Thus,
    \begin{equation}
        \label{eq:median_are_inlier}
        \|r-\Tilde{r}\|_{\cube_n,\infty}    =\max_{\bm j\in[m]^n}|\Tilde{y}_{\bm j}-p(\mathbf{x}_{\bm j}')|\leq\sigma.
    \end{equation}

              Since $p\in{\cal P}_{\bm d}$ and $\hat{p}$ is the minimizer of \eqref{eq:l_infty_minimizer}, we have,
    \begin{equation}
        \label{eq:med_vs_proxy}
        \|\hat{r}-\Tilde{r}\|_{\cube_n,\infty}=\max_{\bm j\in[m]^n}|\hat{p}(\Tilde{\mathbf{x}}_{\bm j})-\Tilde{y}_{\bm j}|\leq \max_{\bm j\in[m]^n}|p(\Tilde{\mathbf{x}}_{\bm j})-\Tilde{y}_{\bm j}|\leq\|p-\Tilde{r}\|_{\cube_n,\infty}
    \end{equation}    
    
    Further, by the triangle inequality, we have
    \begin{align*}
        \|p-\hat{p}\|_{\cube_n,\infty}&\leq \|p-\Tilde{r}\|_{\cube_n,\infty}+\|\Tilde{r}-\hat{r}\|_{\cube_n,\infty}+\|\hat{r}-\hat{p}\|_{\cube_n,\infty}\\
        &\leq 2\|p-\Tilde{r}\|_{\cube_n,\infty}+\eps \|\hat{p}\|_{\cube_n,\infty}\tag{By \eqref{eq:med_vs_proxy} and \autoref{thm:optimal-l-infty-bound} for  $\hat{p}$}\\
        &\leq 2(\|p-r\|_{\cube_n,\infty}+ \|r-\Tilde{r}\|_{\cube_n,\infty})+\eps(\|p-\hat{p}\|_{\cube_n,\infty}+\|p\|_{\cube_n,\infty})\\
        &\leq 3\eps\|p\|_{\cube_n,\infty}+2\sigma +\eps \|p-\hat{p}\|_{\cube_n,\infty}\tag{By \autoref{thm:optimal-l-infty-bound} for $p$ and \eqref{eq:median_are_inlier}}
    \end{align*}

    Rearranging, and using the fact that $\frac{1}{1-\eps}\leq 1+2\eps\leq 2$ for $\eps\leq 1/2$, we conclude:
    $$\|p-\hat{p}\|_{\cube_n,\infty}\leq(2+4\eps)\sigma+6\eps\|p\|_{\cube_n,\infty}.$$
 
    Rescaling $\eps$ to $\eps/6$ gives the desired bound. 
\end{proof}

\section{Deferred Proofs From \autoref{sec:l-infty-with-l-1}}\label{subsec:proof-of-sec:l-1-bound}
In order to prove \autoref{thm:l1-minimizer-l-infty-error}, we need a variation (which is, in fact, a corollary) of \autoref{thm:optimal-l-infty-bound} for $\ell_1$: 
\begin{cor}[$\ell_1$ approximation by piece-wise constant functions]\label{cor:multivariate-l_q-l_1}
 Let $p:\cube\to\R$ be a polynomial of individual degree at most $d$, and $r\colon \cube\to\R$ a piece-wise constant function with respect to the $(m,n)$-Chebyshev partition, such that for all $\bm j\in[m]^n$ there exists $\mathbf{x}^{\bm j}\in \cube_{\bm j}$, for which: $r(\mathbf{x})=p(\mathbf{x}^{\bm j})$, for all $\mathbf{x} \in\cube_{\bm j}$. Then, for some absolute constant $c>1$,
                \[\|p-r\|_{\cube_n,1}\leq \frac{(cd)^{2n+1}}{m}\|p\|_{\cube_n,1}.\]
\end{cor}
Before proving it, we note a useful result, and an observation:
\begin{lem}\label{lem:holders}[H\"older's Inequality]
    Let $\alpha,\beta,\gamma\in\R_{\geq 1}$ such that $\frac{1}{\alpha}+\frac{1}{\beta}=\frac{1}{\gamma}$. For all functions $f$  and $g$ with finite $\|f\|_{S,\alpha}$, and $\|g\|_{S,\beta}$, we have: $\|fg\|_{S,\gamma}\leq\|f\|_{S,\alpha}\|g\|_{S,\beta}$.   
\end{lem}
\begin{obs}\label{obs:norms-equiv}[Equivalence of norms]
    Let $1\leq \gamma<\alpha$, and $S\subseteq\cube$. Define $\frac{1}{\beta}\triangleq\frac{1}{\gamma}-\frac{1}{\alpha}$, and $g(\mathbf{x})\triangleq 1\text{, for all }\mathbf{x}\in\cube$. Then, for any $f:\cube\to\R$ with finite $\|f\|_{S,\alpha}$, by \autoref{lem:holders}, we have \[\|f\|_{S,\gamma}\leq\|f\|_{S,\alpha}\|\bm 1\|_{S,\beta}=V^{\frac{1}{\beta}}_n(S)\|f\|_{S,\alpha}\leq 2^{\frac{n}{\beta}}\|f\|_{S,\alpha}=2^{\frac{n}{\gamma}-\frac{n}{\alpha}}\|f\|_{S,\alpha}.\]
    In particular, in the limit of $\alpha \to \infty$, 
    we have

        \begin{equation}\label{eq:holder_for_infinity}
        \|f\|_{S,\gamma}\leq V_n^{\frac{1}{\gamma}}(S)\|f\|_{S,\infty}\le 2^{\frac{n}{\gamma}}\|f\|_{S,\infty}.
    \end{equation}
\end{obs}

\begin{proof}[Proof of \autoref{cor:multivariate-l_q-l_1}]
Observe,    
 \[\|p-r\|_{\cube_n,1}\leq 2^n\|p-r\|_{\cube_n,\infty}\leq O\left(\frac{dn}{m}\right)2^n\|p\|_{\cube_n,\infty}\leq O\left(\frac{dn}{m}\right)(2\sqrt{2}d)^{2n}\|p\|_{\cube_n,1}.\]
 The first inequality is from  \autoref{obs:norms-equiv}\eqref{eq:holder_for_infinity}  with $\gamma=1$, 
    the second is by  \autoref{thm:optimal-l-infty-bound},  and the last by  \autoref{thm:improved-sandwich}.
\end{proof}

We next generalize \cite[Lemma 3.1]{kkp}. We show that on a fine enough Chebyshev grid, an empirical estimate of the $\ell_1$ norm of $p$ is close to the actual value, thus implying it can be used as a proxy for the actual:

\begin{thm}[Empirical $\ell_1$ estimate suffices]\label{lem:mult_empirical_l1}
    Let   $p:\cube\to\R$ be a polynomial of individual degree at most $d$, and $m\geq(cd)^{2n+1}/\eps$, for large enough constant $c>1$.     Given a set of $M$ samples $    (\mathbf{x}_i,y_i)    $, such that for every $\bm j\in[m]^n$, the set of samples in the  cell $\cube_{\bm j}$, denoted by $S_{\bm j}\triangleq\{\alpha\in[M]:\mathbf{x}_{\alpha}\in \cube_{\bm j}\}$ is not empty. 
    Define   the weighted average of $p$ with respect to the the $(m,n)$-Chebyshev partition:  
    \[\Phi(p)\triangleq \sum_{\bm j\in[m]^n}\frac{V_n(\cube_{\bm j})}{|S_{\bm j}|}\sum_{\alpha\in S_{\bm j}}|p(\mathbf{x}_{\alpha})|.\]
    Then, $\Phi(p) \in(1\pm\eps)\|p\|_{\cube_n,1}$.
\end{thm}

\begin{proof} 

Fix some $\bm j\in [m]^n$. Denote the average value of samples in the cell $\cube_{\bm j}$ by  $r_{\bm j}\triangleq\sum_{\alpha\in S_{\bm j}}\frac{|p(\mathbf{x}_{\alpha})|}{|S_{\bm j}|}$. 
Define 
\[a_{\bm j}\triangleq\min_{\mathbf{x}\in\cube_{\bm j}}\{|p(\mathbf{x})|\} \text{,\qquad and\qquad } b_{\bm j}\triangleq\max_{\mathbf{x}\in\cube_{\bm j}}\{|p(\mathbf{x})|\}.\]
Observe $r_{\bm j}\in[a_{\bm j},b_{\bm j}]$. 
By continuity of $|p|$, there exists $\bfx^{(\bmj)}\in\cube_{\bm j}\text{, such that }|p(\mathbf{x}^{(\bm j)})|=r_{\bm j}$.

Let $r\colon\cube\to\R$ be a piece-wise constant function with respect to the $(m,n)$-Chebyshev partition, which is defined as:  
$r(\mathbf{x})=p(\mathbf{x}^{\bm j})$, 
for all 
$\mathbf{x}\in \cube_{\bm j}$.
Observe,
    \begin{align*}
        \Phi(p) &                =\sum_{\bm j\in[m]^n}V_n(\cube_{\bm j})r_{\bm j}   =\int_{\cube}|r(\mathbf{x})| d\mathbf{x}=\|r\|_{\cube_n,1}.    \end{align*}
    By triangle inequality, we have
    \begin{align*}
        \|p\|_{\cube_n,1}-\|r-p\|_{\cube_n,1}\leq\Phi(p) =\|r\|_{\cube_n,1}&\leq\|r-p\|_{\cube_n,1}+\|p\|_{\cube_n,1}
        \end{align*}
        By our choice of $m$, and \autoref{cor:multivariate-l_q-l_1}, we have $\|p-r\|_{\cube_n,1}\leq\eps\|p\|_{\cube_n,1}$. So, we conclude
        \begin{align*}
        (1-\eps)\|p\|_{\cube_n,1}\leq\Phi(p)&\leq(1+\eps)\|p\|_{\cube_n,1}.\qedhere
    \end{align*}
\end{proof}

Next, we define the \emph{closeness} of a set of samples to $p$, in $\ell_1$, and relate it to our notion of \emph{goodness}.
We show our $\ell_1$ minimizer outputs a close approximation, if the samples are good (and, hence close to $p$ in $\ell_1$). 
\begin{defn}[$(\alpha,\gamma)$-close in $\ell_1$]\label{defn:l1_close-sample}
Let $p\colon\cube\to\R$ be a polynomial of individual degree at most $d$. 
    For a set of $M$ samples $S\triangleq\{(\mathbf{x}_i,y_i)\}$,     consider the  $(m,n)$-Chebyshev partition. 
    Let  $S_{\bm j}\triangleq\{\beta\in[M]:\mathbf{x}_{\beta}\in \cube_{\bmj}\}$ be the set of samples that are in the cell $\cube_{\bmj}$.  
    Let $\alpha<1/2$, and $\gamma>0$. 
    For every $\bmj\in[m]^n$, let 
     \[e_{\bmj}\triangleq \min_{\substack{ S'\subset S_{\bm j},\\ |S'|\leq \lceil(1-\alpha)|S_{\bm j}|\rceil}} \max_{\beta\in S'}|p(\mathbf{x}_{\beta})-y_{\beta}|.\]
    We say that $S$ is $(\alpha,\gamma)$-close to $p$  in $\ell_1$ with respect to the partition, if $|S_{\bmj}|\geq 1/\alpha$ for all $\bm j\in [m]^n$ and, in addition, 
            \begin{equation}
        \label{eq:ell_1_closeness}
       \sum_{\bmj\in[m]^n}V_n(\cube_{\bm j})e_{\bm j}\leq\gamma. 
    \end{equation}
    
\end{defn}
For every $\bm j\in[m]^n$, let $S'_{\bm j}\subseteq S_{\bm j}$ be the set of inliers in the cell $\cube_{\bm j}$. 
  If a set of   samples $S$   is $\alpha$-good,   then  the fraction of outliers in $S_{\bm j}$ is less than $\alpha$, and hence   $|S'_{\bm j}|>(1-\alpha)|S_{\bm j}|$. 
  Note that for every $\beta\in S'_{\bm j},|p(\mathbf{x}_{\beta})-y_{\beta}|\leq\sigma$, 
  making $e_{\bm j}\leq\sigma$. 
  Hence $\sum_{\bm j\in[m]^n}V_n(\cube_{\bm j})e_{\bm j}\leq\sigma V_n(\cube)=2^n\sigma$.
   We conclude:\begin{obs}[$\alpha$-good $\implies(\alpha,2^n\sigma)$-close]\label{obs:alpha-good-to-closeness}
    If a set $S$ of samples     is $\alpha$-good for the $(m,n)$-Chebyshev partition,     then $S$ is also $(\alpha,2^n\sigma)$-close to $p$ in $\ell_1$.
\end{obs}

Finally, we  bound the $\ell_1$ regression error of the $\ell_1$ minimizer $\hat{p}_{\ell_1}$, with respect to  $p$. 
The minimizer $\hat{p}_{\ell_1}$ is assumed to be computed on a set of samples $S$, that is $(\alpha,\gamma)$-close to $p$ in $\ell_1$:

\begin{lem}[$\ell_1$ error bound for the $\ell_1$ minimizer]\label{lem:mult-l1-error-bound}

    Let $\alpha<1/2,\eps\leq(1-2\alpha)/2$, and $m\geq(cd)^{2n+1}/\eps$, for some large enough constant $c>1$. 
    Given a set of $M$ samples $S\triangleq\{(\mathbf{x}_i,y_i)\}$, that is $(\alpha,\gamma)$-close to an individual degree-$d$ polynomial $p\colon\cube\to\R$ in $\ell_{1}$ with respect to the $(m,n)$-Chebyshev grid, if $\hat{p}_{\ell_1}$ is the $\ell_1$ minimizer from \autoref{defn:l1-minimzer},
                    then
    \[\|p-\hat{p}_{\ell_1}\|_{\cube_n,1}\leq\frac{4\gamma}{1-2\alpha}.\]
        \end{lem}
\begin{proof}
    In each cell $\cube_{\bm j}$, call the $\lfloor\alpha|S_{\bm j}|\rfloor$ points that maximize $|p(\mathbf{x}_{\beta})-y_{\beta}|$ the  \emph{bad} points, and the rest are \emph{good}, denoted by $B_{\bm j}$, and $G_{\bm j}$ respectively. Let the objective function be defined as 
    \[obj(f)\triangleq \sum_{\bm j\in[m]^n}\frac{V_n(\cube_{\bm j})}{|S_{\bm j}|}\sum_{\beta\in S_{\bm j}}|f(\mathbf{x}_{\beta})-y_{\beta}|,\] and with $\hat{p}_{\ell_1}$ as its minimizer, we have     \begin{align*}
        0&\geq obj(\hat{p}_{\ell_1})-obj(p)=\sum_{\bm j\in[m]^n}\frac{V_n(\cube_{\bm j})}{|S_{\bm j}|}\sum_{\beta\in S_{\bm j}}\left(|\hat{p}_{\ell_1}(\mathbf{x}_{\beta})-y_{\beta}|-|p(\mathbf{x}_{\beta})-y_{\beta}|\right)=(\star).
        \end{align*}
        Now, using the triangle inequality, for the \emph{good} samples, we bound $$|\hat{p}_{\ell_1}(\mathbf{x}_{\beta})-y_{\beta}|\geq |\hat{p}_{\ell_1}(\mathbf{x}_{\beta})-p(\mathbf{x}_{\beta})|-|p(\mathbf{x}_{\beta})-y_{\beta}|.$$ Whereas for the \emph{bad} samples, we bound $$|\hat{p}_{\ell_1}(\mathbf{x}_{\beta})-y_{\beta}|-|p(\mathbf{x}_{\beta})-y_{\beta}|\geq -|p(\mathbf{x}_{\beta})-\hat{p}_{\ell_1}(\mathbf{x}_{\beta})|.$$ 
        Therefore,
        \begin{align*}
            (\star)\geq &\sum_{\bm j\in[m]^n}\!\!\frac{V_n(\cube_{\bm j})}{|S_{\bm j}|}\sum_{\beta\in G_{\bm j}} (|\hat{p}_{\ell_1}(\mathbf{x}_{\beta})-p(\mathbf{x}_{\beta})|-2|p(\mathbf{x}_{\beta})-y_{\beta}|) -\!\!\sum_{\bm j\in[m]^n}\!\!\frac{V_n(\cube_{\bm j})}{|S_{\bm j}|}\sum_{\beta\in B_{\bm j}}|\hat{p}_{\ell_1}(\mathbf{x}_{\beta})-p(\mathbf{x}_{\beta})|\\
            \geq &\underbrace{\!\!\sum_{\bm j\in[m]^n}\!\!\!\! \frac{V_n(\cube_{\bm j})}{|S_{\bm j}|}\!\sum_{\beta\in G_{\bm j}} |\hat{p}_{\ell_1}(\mathbf{x}_{\beta})-p(\mathbf{x}_{\beta})|}_{(I)} -
            \underbrace{\!\!\sum_{\bm j\in[m]^n}\!\!\!\! \frac{V_n(\cube_{\bm j})}{|S_{\bm j}|}\!\!\sum_{\beta\in B_{\bm j}}|\hat{p}_{\ell_1}(\mathbf{x}_{\beta})-p(\mathbf{x}_{\beta})|}_{(II)}
            -2\underbrace{\!\!\sum_{\bm j\in[m]^n}\!\!\!\! \frac{V_n(\cube_{\bm j})}{|S_{\bm j}|}|G_{\bm j}|e_{\bm j}}_{(III)},
            \end{align*}
            where last inequality follows, since $|p(\mathbf{x}_{\beta})-y_{\beta}|\leq e_{\bm j }$ for all \emph{good} samples.
            Next, we bound each term separately.
            The third term is  $(III)\leq\gamma$ follows from \eqref{eq:ell_1_closeness} and the fact that $G_{\bm j}\subseteq S_{\bm j}$.

            For the first and second terms, we use the fact that for every $\bm j\in[m]^n$, $|B_{\bm j}|\leq \alpha |S_{\bm j}|$ and $|G_{\bm j}|\geq (1-\alpha)|S_{\bm j}|$, we note that since $|S_{\bm j}|\geq 1/\alpha$, both $B_{\bm j}$ and $G_{\bm j}$ are not empty, and we can apply \autoref{lem:mult_empirical_l1}. 
            $$(I)\geq (1-\alpha)\sum_{\bm j\in[m]^n}\frac{V_n(\cube_{\bm j})}{|G_{\bm j}|}\sum_{\beta\in G_{\bm j}} |\hat{p}_{\ell_1}(\mathbf{x}_{\beta})-p(\mathbf{x}_{\beta})| \geq (1-\alpha)(1-\eps)\|\hat{p}_{\ell_1}-p\|_{\cube_n,1},$$
            and 
            $$(II)\leq \alpha\sum_{\bm j\in[m]^n}\frac{V_n(\cube_{\bm j})}{|B_{\bm j}|}\sum_{\beta\in B_{\bm j}} |\hat{p}_{\ell_1}(\mathbf{x}_{\beta})-p(\mathbf{x}_{\beta})| \leq \alpha(1+\eps)\|\hat{p}_{\ell_1}-p\|_{\cube_n,1}.$$
                                                                                  
    Combining the bounds and rearranging we conclude: 
    \[\|p-\hat{p}_{\ell_1}\|_{\cube_n,1}\leq\frac{2\gamma}{1-2\alpha-\eps}\leq\frac{4\gamma}{1-2\alpha}.\]
    The last inequality follows from the assumption
  $\eps\leq(1-2\alpha)/2$.   
\end{proof}

As an immediate corollary, we can bound the $\ell_\infty$ error  of the $\ell_1$ minimizer $\hat{p}_{\ell_1}$ on a set $S$ of samples that is $\alpha$-good, thus proving the main theorem of this subsection. 
\begin{proof}[Proof of \autoref{thm:l1-minimizer-l-infty-error}]
    Since $S$ is assumed to be $\alpha$-good, by \autoref{obs:alpha-good-to-closeness}, it is $(\alpha,2^n\sigma)$-close to $p$ in $\ell_1$. So, invoking \autoref{lem:mult-l1-error-bound}, we get $\|p-\hat{p}_{\ell_1}\|_{\cube_n,1}\leq 2^nO_{\alpha}(\sigma)$. Further, since both $p$, and $\hat{p}_{\ell_1}$ are polynomials of individual degree at most $d$, so is their difference. Hence, invoking \autoref{thm:improved-sandwich}, we get $\|p-\hat{p}_{\ell_1}\|_{\cube_n,\infty}\leq (4d^2)^n\|p-\hat{p}_{\ell_1}\|_{\cube_n,1}\leq (8d^2)^n O_{\alpha}(\sigma)$.
\end{proof}

\section{Weaker \texorpdfstring{$\ell_\infty$ to $\ell_1$}{l-infty-l-one} norms relation
}\label{sec:wilhwlmwnsen}

This section is devoted to proving \autoref{thm:sandwich}, a weaker version of \autoref{thm:improved-sandwich} for \emph{total} degree-$d$ polynomials, the proof of which relies on a result by Wilhelmsen\cite{WILHELMSEN1974216} and avoids the inductive argument of \autoref{thm:improved-sandwich}. 
Using this weaker bound instead would lead to sample complexity of $O(m^n\log m^n)\geq ((c_0d)^{2n+2}n^{n/2})^n\log d$ instead of $(cd)^{n(2n+1)}\log d$, for some absolute constants $c,c_0>0$.
 
 Let $T$ be a compact, convex subset of $\R^n$, with boundary $\partial T$, and interior $T^0\neq\emptyset$. Fix a $\bm t_0\in\partial T$ and a unit vector $\bm u\in\R^n$. Consider the hyperplane with normal $\bm u$
  \[\mathcal{H}_{\bm u}\triangleq\{\bm t\in\R^n:\langle \bm t-\bm t_0,\bm u\rangle=0\}.\]
 
 $\mathcal{H}_{\bm u}$ is a support hyperplane  of $T$ at $t_0$ if is not containing any $\bm t\in T^0$. In that case, $\bm u$ is called an outer normal to $T$ at $t_0$.
    For any direction $\bm u$,     there exist precisely two support hyperplanes of $T$, with outer normals $\bm u$, and $-\bm u$. They are separated by a distance $\lambda_{\bm u}>0$. The width of $T$ is defined (\cite[Definition 2.1]{WILHELMSEN1974216}) to be $\omega_T\triangleq\min_{\|\bm u\|_2=1}\lambda_{\bm u}$.
We will use the following result by Wilhelmsen:
\begin{thm}\protect{\cite[Theorem 3.1]{WILHELMSEN1974216}}\label{thm:wilhelmsen}
   Let $T$ be a compact, convex subset of $\R^n$. For any \emph{total} degree-$d$ polynomial $p:T\to\R$, we have\footnote{Here, for any $\bm v\in\R^n$, $\|\bm v\|$ denotes $\|\bm v\|_2$, the $\ell_2$ norm of $\bm v$, and $\|\bm v\|_{T,\infty}$ denotes $\max_{\bm v\in T}\{\|\bm v\|_2\}$.}
    \[\|\grad p\|_{T,\infty}\leq\frac{4d^2}{\omega_T}\|p\|_{T,\infty}.\]
\end{thm}
Observe, 
the width of $\cube=[-1,1]^n$ is $\omega_{\cube}=2$. So, we have
\begin{obs}\label{obs:wilhemlsen-useful}[Wilhelmsen on the cube $\cube_n$]
\[\|\grad p\|_{\cube_n,\infty}\leq 2d^2\|p\|_{\cube_n,\infty}
\]
\end{obs}
Using this observation, we prove a weaker version of \autoref{thm:improved-sandwich}:
\begin{thm}\label{thm:sandwich}[Weaker $\ell_\infty$-$\ell_1$ norms relation]
    Let $p:\cube\to\R$ be a polynomial of \emph{total} degree-$d$. Then, 
    \[
        \|p\|_{\cube_n,\infty}\leq 2\sqrt{n\pi}\left(4d^2\sqrt{\frac{2n}{e\pi}}\right)^{n}\|p\|_{\cube_n,1}.\]
\end{thm}
\begin{proof}
                Let $\mathbf{x}^*\in\cube:p(\mathbf{x}^*)=\|p\|_{\cube_n,\infty}$. Define a set of points close enough to $\mathbf{x}^*$, i.e. $S(\mathbf{x}^*)\triangleq\{\bm y\in\cube:\|\mathbf{x}^*-\bm y\|_2\leq 1/4d^2\}\subseteq\cube$. Then, we lower bound $p$ on all points in $S(\mathbf{x}^*)$. Formally:
    \begin{lem}\label{lem:local-sandwich}
        For every $\bm y\in S(\mathbf{x}^*),|p(\bm y)|\geq |p(\mathbf{x}^*)|/2$.
    \end{lem}
    \begin{proof}
        Fix a $\bm y\in S(\mathbf{x}^*)$. Let $\pline_{\mathbf{x}^*,\bm y}$ be the line segment connecting $\mathbf{x}^*$ and $\bm y$. Then, by \autoref{thm:mean-value}, there exists a $\bm z\in\pline_{\mathbf{x}^*,\bm y-\mathbf{x}^*}$, such that $p(\mathbf{x}^*)-p(\bm y)=\langle\grad p(\bm z),\mathbf{x}^*-\bm y\rangle$.
        \begin{align*}
            \implies|p(\mathbf{x}^*)-p(\bm y)|&=|\langle\grad p(\bm z),\mathbf{x}^*-\bm y\rangle|\leq\|\grad p(\bm z)\|_2\cdot\|\mathbf{x}^*-\bm y\|_2 \tag{By Cauchy-Schwarz}\\
            &\leq \|\grad p\|_{\cube_n,\infty}\cdot\underbrace{\|\mathbf{x}^*-\bm y\|_2}_{\leq 1/4d^2,\because\bm y\in S(\mathbf{x}^*)}\\
            &\leq 2d^2\|p\|_{\cube_n,\infty}/(4d^2)\tag{By \autoref{obs:wilhemlsen-useful}}\\
            &=\frac{\|p\|_{\cube_n,\infty}}{2}=\frac{|p(\mathbf{x}^*)|}{2}.
        \end{align*}
                $\implies|p(\bm y)|\geq |p(\mathbf{x}^*)|-|p(\mathbf{x}^*)-p(\bm y)|\geq |p(\mathbf{x}^*)|/2$
    \end{proof}
    Now observe,
    \begin{align*}
        \|p\|_{\cube_n,1}&=\int_{\cube}|p(\mathbf{y})| d\mathbf{y}\geq\int_{S(\mathbf{x}^*)\cap\cube}|p(\bm y)| d\bm y \geq\int_{S(\mathbf{x}^*)\cap\cube}\frac{|p(\mathbf{x}^*)|}{2} d\bm y \tag{By \autoref{lem:local-sandwich}}\\
        &\geq\frac{1}{2^n}\int_{S(\mathbf{x}^*)}\frac{|p(\mathbf{x}^*)|}{2}d\mathbf{x} \tag{ $S(\mathbf{x}^*)\cap\cube$ covers at least one orthant of $S(\mathbf{x}^*)$}\\
        &=\frac{|p(\mathbf{x}^*)|}{2^{n+1}} V_n(S(\mathbf{x}^*))=\frac{|p(\mathbf{x}^*)|}{2^{n+1}}\cdot\frac{\pi^{n/2}}{\Gamma(n/2+1)(4d^2)^n}\tag{$S(\mathbf{x}^*)$ is an $n$-ball of radius $1/4d^2$}\\
        &\approx\frac{\|p\|_{\cube_n,\infty}}{2^{n+1}\sqrt{n\pi}}\left(\frac{1}{4d^2}\sqrt{\frac{2e\pi}{n}}\right)^n.\tag{By Stirling's approximation of the Gamma function}
    \end{align*}
    \[\implies\|p\|_{\cube_n,\infty}\leq 2\sqrt{n\pi}\left(4d^2\sqrt{\frac{2n}{e\pi}}\right)^n\|p\|_{\cube_n,1}.\qedhere\]
\end{proof}
If we use \autoref{thm:sandwich}, instead of \autoref{thm:improved-sandwich}, we get a weaker form of \autoref{cor:multivariate-l_q-l_1}:
\begin{cor}\label{cor:multivariate-l_q-l_1-wilhelmsen}
 Let $p:\cube\to\R$ be a polynomial of \emph{total} degree at most $d$, and $r\colon \cube\to\R$ a piece-wise constant function with respect to the $(m,n)$-Chebyshev partition, such that for all $\bm j\in[m]^n$ there exists $\mathbf{x}^{\bm j}\in \cube_{\bm j}$, for which $r(\mathbf{x})=p(\mathbf{x}^{\bm j})$. Then
                \[\|p-r\|_{\cube_n,1}\leq O\left(\frac{d^2n\sqrt{n}}{m}(16d^2\sqrt{n})^n\right)\|p\|_{\cube_n,1}.\]
\end{cor}
Using this, we get that, for $m\geq\frac{cd^2n\sqrt{n}}{\eps}(16d^2\sqrt{n})^n$, for some absolute constant $c>0$, the $\ell_\infty$ error of the $\ell_1$ minimizer $\hat{p}_{\ell_1}$, with respect to the $(m,n)$-Chebyshev grid, can be bounded by $poly(d^n)\sigma$, via a worse form of \autoref{thm:l1-minimizer-l-infty-error}:
\begin{cor}\label{cor:l1-minimizer-l-infty-error-wilhelmsen}
    Let $\alpha<0.5$ be constant $,\eps\leq(1-2\alpha)/2$, and $m\geq\frac{c_0^nd^{2n+2}n^{n/2}}{\eps} \ge \frac{cd^2n\sqrt{n}}{\eps}(16d^2\sqrt{n})^n$, for some constants $c>0,c_0>0$. Given a set $S$ of     that is $\alpha$-good with respect to the $(m,n)$-Chebyshev partition, 
    with $\hat{p}_{\ell_1}$ as in \autoref{defn:l1-minimzer}, we have\footnote{Here $p$, and $\hat{p}_{\ell_1}$ are polynomials of \emph{total} degree at most $d$.}
    \[\|p-\hat{p}_{\ell_1}\|_{\cube_n,\infty}\leq O((8d^2)^n \sigma).\]
\end{cor}
Since we use $\ell_1$ minimizer, only in the case of arbitrarily large $\|p\|_{\cube_n,\infty}$ (\autoref{sec:l-infty-with-l-1}), this worsens only the sample complexity of \autoref{alg:median_rec_with_l1}.

\end{document}